\documentclass[12pt,onecolumn,twoside]{IEEEtran}

\usepackage{calc}

\usepackage{multirow}
\usepackage{cite}
\usepackage{graphicx,subfigure}
\usepackage{psfrag}
\usepackage{amsmath,amssymb,amsthm}
\usepackage{color}
\usepackage{tikz} 
\usepackage{bm}
\usepackage{bbm}
\usepackage{verbatim}

\usepackage{cases}

\DeclareMathOperator*{\defeq}{\triangleq}

\newtheorem{theorem}{Theorem}
\newtheorem{corollary}{Corollary}

\newtheorem{lemma}{Lemma}

\newcommand{\bit}{\begin{itemize}}
\newcommand{\eit}{\end{itemize}}

\newcommand{\bc}{\begin{center}}
\newcommand{\ec}{\end{center}}

\newcommand{\ba}{\begin{array}}
\newcommand{\ea}{\end{array}}

\newcommand{\beq}{\begin{equation}}
\newcommand{\eeq}{\end{equation}}

\newcommand{\beqn}{\begin{equation*}}
\newcommand{\eeqn}{\end{equation*}}

\newcommand{\bean}{\begin{eqnarray*}}
\newcommand{\eean}{\end{eqnarray*}}
\newcommand{\bea}{\begin{eqnarray}}
\newcommand{\eea}{\end{eqnarray}}

\def\E{\mathbb{E}}

\def\bv{\boldsymbol{b}}

\def\vv{\boldsymbol{v}}

\newcommand{\Ac}{{\mathcal A}}

\newcommand{\Ic}{{\mathcal I}}

\newcommand{\Rc}{{\mathcal R}}

\newcommand{\Zc}{{\mathcal Z}}
\newcommand{\T}{{\scriptscriptstyle\mathsf{T}}}

\newtheorem{remark}{Remark}

\newcommand{\dsum}{d_{\text{sum}}}

\newtheorem{claim}{Claim}

\newcommand{\non}{\nonumber}

\newcommand{\Hen}{\mathbb{H}}
\newcommand{\hen}{\mathrm{h}}
\newcommand{\Imu}{\mathbb{I}}

\newcommand{\bln}{n}

\DeclareMathOperator*{\argmax}{arg\,max}

\usepackage{mathtools}
\usepackage{xifthen}

\newcommand{\al}{\boldsymbol{\alpha}}

\newcommand{\Jl}{J_m}
\newcommand{\Jo}{J_0}
\newcommand{\Nd}{N}
\newcommand{\Md}{M}
\newcommand{\me}{m}
\newcommand{\Ke}{K}

\begin{document}
\sloppy

\title{Multi-layer Interference Alignment and GDoF of the $K$-User Asymmetric Interference Channel}

\author{Jinyuan Chen 
\thanks{Jinyuan Chen is with Louisiana Tech University, Department of Electrical Engineering, Ruston, USA (email: jinyuan@latech.edu).} 
}

\maketitle
\pagestyle{headings}

\begin{abstract}

In  wireless networks, link strengths are often affected by some topological factors such as propagation path loss, shadowing and inter-cell interference. Thus, different users in the network might experience different link strengths.
In this work we consider a $K$-user asymmetric  interference channel, where the channel gains of the links  connected to Receiver~$k$ are scaled with $\sqrt{P^{\alpha_k}}$, $k=1,2, \cdots, K$, for  $0< \alpha_1 \leq   \alpha_2 \leq \cdots \leq  \alpha_K \leq1$.
For this setting, we show that the optimal sum \emph{generalized} degrees-of-freedom (GDoF) is characterized as \[\dsum = \frac{ \sum_{k=1}^K \alpha_k  +  \alpha_K -\alpha_{K-1}}{2}\] 
which matches the existing result $\dsum= \frac{K}{2}$ when $\alpha_1 =   \alpha_2 = \cdots =  \alpha_K =1$.
The achievability is based on  multi-layer interference alignment, where different interference alignment sub-schemes are designed in different  layers associated with specific power levels, and  successive decoding is applied at the receivers.
While the converse for the \emph{symmetric} case only requires bounding the sum degrees-of-freedom (DoF) for selected \emph{two} users, the converse for  this \emph{asymmetric} case involves bounding the \emph{weighted} sum GDoF for selected $J+2$ users, with corresponding weights  $(2^{J}, 2^{J-1}, \cdots, 2^{2}, 2^{1})$, a geometric sequence with common ratio 2, for the first $J$ users and  with corresponding weights $(1, 1)$ for the last two users, for $J \in \{1,2, \cdots, \lceil \log \frac{K}{2} \rceil\}$.

\end{abstract}

\section{Introduction}

In wireless networks, the strengths of communication links  are often affected by  propagation path loss, shadowing, inter-cell interference, and some other topological factors. Therefore, different users in the network might experience different link strengths. 
For example, in an interference network, when a receiver is relatively far from the transmitters, this receiver might experience weaker links  compared to the receivers that are more close to the transmitters (see Fig.~\ref{fig:ICa}). 
Such asymmetry property of the link strengths in communication networks can crucially affect the transceiver design, as well as the capacity performance.  

In this work we consider a $K$-user asymmetric interference channel, where different receivers might have different link strengths.
For this setting, the channel gains of the links  connected to Receiver~$k$ are scaled with $\sqrt{P^{\alpha_k}}$, where $\alpha_k$ captures the  \emph{link strength} of Receiver~$k$,  which might be different from that of the other receivers, for  $k=1,2, \cdots, K$. 
This generalizes  the symmetric setting, in which $\alpha_1 =   \alpha_2 = \cdots =  \alpha_K =1$, to a  setting with  diverse link strengths. 

For the symmetric $K$-user interference channel, the work in \cite{CJ:08} showed that the optimal sum degrees-of-freedom (DoF) is characterized by $K/2$, which implies that ``everyone gets half of the cake''. DoF is a pre-log factor of capacity at the  high signal-to-noise ratio (SNR) regime.  
Although  the DoF metric can produce profound insights, it has a fundamental limitation, that is, it  treats  all non-zero links as approximately
equally strong. 
Thus, it motivates the researchers to go  beyond the DoF metric  into the generalized degrees-of-freedom (GDoF) metric (see \cite{ETW:08,VKV:11,KV:11,KV:12,KV:12it,MV:18,BLK:13,HCJ:12,CEJ:14isit,MTP:13, TMP:13, JV:10, CEJ:15,GNAJ:15, ChenAllerton:18,GTJ:15, CLtifs:19,YC:16, CG:19isit,MM:11,LC:19isit,SJ:15,DYJ:18,DJ:19,WYHJ:19} and the references therein), for the settings with diverse link strengths.    
For the $K$-user asymmetric interference channel, we  focus on the optimal sum GDoF. Specifically, for this asymmetric setting we show that the optimal sum GDoF is characterized as $\dsum = \frac{ \sum_{k=1}^K \alpha_k  +  \alpha_K -\alpha_{K-1}}{2}$,  
for  $0< \alpha_1 \leq   \alpha_2 \leq \cdots \leq  \alpha_K \leq1$. 
This result generalizes the  existing result of the symmetric case to the setting with diverse link strengths.

The proposed achievability is based on  multi-layer interference alignment and successive decoding. 
While the traditional  interference alignment scheme is usually dedicated to all users in the network  (cf.~\cite{CJ:08,MMK:08ia}), 
the multi-layer interference alignment scheme proposed in this work consists of $K$  different interference alignment sub-schemes,  with each interference alignment sub-scheme dedicated to a subset of users. In this scheme, each  interference alignment sub-scheme is  designed in a specific  layer associated with a particular power level. In terms of decoding, successive decoding is applied at the receivers. 
Specifically, successive decoding is operated  layer by layer. For the decoding at one layer, each of the involved receivers decodes the desired signals  and the interference in this layer, and then remove them to  decode signals at the next layer. 
The converse for  this \emph{asymmetric} case involves bounding the \emph{weighted} sum GDoF for selected $J+2$ users, with weights being  a geometric sequence for the first $J$ users, for $J \in \{1,2, \cdots, \lceil \log \frac{K}{2} \rceil\}$.
This is very different from the converse for the \emph{symmetric} case, which only requires bounding the sum DoF for selected \emph{two} users.

\begin{figure}
\centering
\includegraphics[width=12cm]{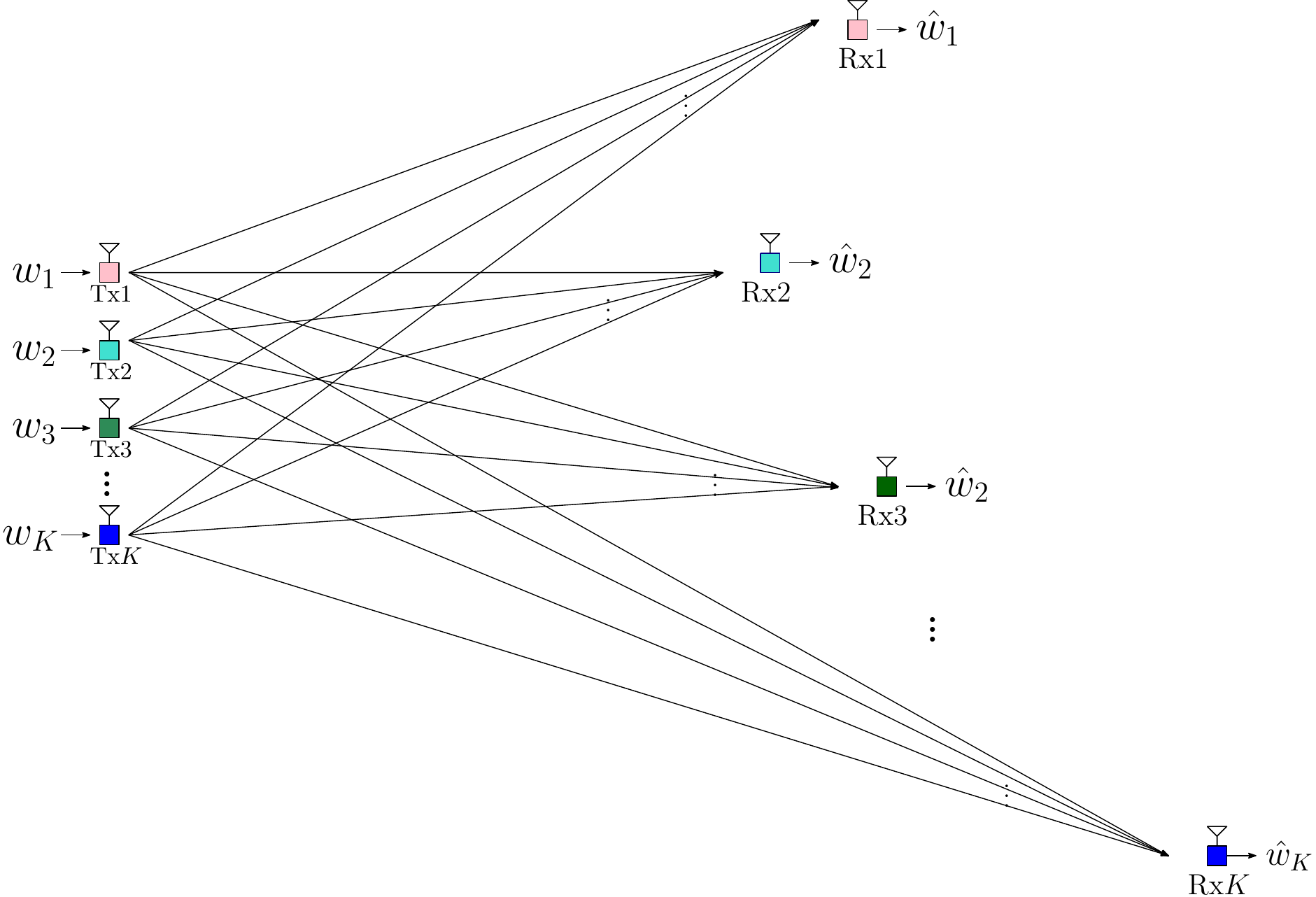}
\vspace{-5pt}
\caption{An asymmetric interference channel, where some receivers are relatively far from the transmitters and  consequently might have weaker links compared to the receivers closer to the transmitters.}
\label{fig:ICa}
\vspace{-10pt}
\end{figure}

The remainder of this work is organized as follows.  
Section~\ref{sec:system} describes the  system model of the asymmetric interference channel.  Section~\ref{sec:mainresult} provides the  main result of this work. The converse proof is provided in Section~\ref{sec:converse}, while the achievability proof is described in  Section~\ref{sec:achievability}. Finally, section~\ref{sec:conclusion} shows the  conclusion of this work. 
Throughout this work, $\Hen(\bullet)$, $\hen(\bullet)$ and $\Imu(\bullet)$ denote the entropy, differential entropy and mutual information,  respectively.  
$|\bullet|$ denotes the magnitude of a scalar or the cardinality of a set. 
     $\Zc$,  $\Zc^+$, $\Rc$ and $\mathbb{N}$ denote the sets of integers, positive integers, real numbers, and natural numbers, respectively.    
   $o(\bullet)$ is a standard Landau notation, where  $f(x)=o(g(x))$ implies that $\lim_{x \to \infty} f(x)/g(x) =0$. 
    $[A: B]$ is a set of integers from $A$ to $B$, for some integers $A\leq B$. 
    Given a set $\Ac$, then $\Ac (i)$ denotes the $i$th element of set $\Ac$. 
Logarithms are in base~$2$.

\section{System model \label{sec:system} }

We focus on a $K$-user asymmetric interference channel defined by the following input-output equations:
\begin{align}
y_{k}(t) &= \sqrt{P^{\alpha_{k}}} \sum_{\ell=1}^{K}   h_{k\ell} x_{\ell}(t) +z_{k}(t), \quad k \in [1: K] \label{eq:channelG} 
\end{align}
$t \in [1: \bln] $, where  $x_{\ell}(t)$ is the channel input at Transmitter~$\ell$  subject to a normalized average power constraint  $\E |x_{\ell}(t)|^2 \leq 1$. $z_k(t) \sim \mathcal{N}(0, 1)$ is additive white Gaussian noise at Receiver~$k$. $h_{k\ell}$ is the channel coefficient between Transmitter~$\ell$ and Receiver~$k$.  $P\geq 1$ denotes a nominal power value. 
The  exponent  $\alpha_{k}$ represents the \emph{channel strength} of the links connected to Receiver~$k$. Without loss of generality we consider the case that \[0< \alpha_1 \leq   \alpha_2 \leq \cdots \leq  \alpha_K \leq1.\] 
The channel coefficients $\{h_{k\ell}\}_{k, \ell}$ are drawn independently and identically from a continuous distribution.  
We assume that the absolute value of each channel coefficient is bounded between a finite maximum value and a nonzero
minimum value. All the channel parameters $\{\alpha_{k}\}_{k}$ and  coefficients $\{h_{k\ell}\}_{k, \ell}$ are assumed to be perfectly known to all the transmitters and receivers (perfect CSIT and CSIR).

In this channel, the message $w_k$ is sent from Transmitter~$k$ to Receiver~$k$ over $\bln$ channel uses, for $k\in [1:K]$, where $w_k$ is  uniformly drawn from a set $\mathcal{W}_k=[1 : 2^{\bln R_k}]$ and  $R_k$ is  the rate of this message.  
A  rate tuple $(R_1(P, \al), R_2 (P, \al), \cdots, R_K (P, \al))$ is said to be achievable  if for any $\epsilon >0$ there exists a sequence of $\bln$-length codes such that each receiver can decode its own message reliably, i.e., $\text{Pr}[ \hat{w_k} \neq w_{k}  ]  \leq \epsilon$,   $ \forall k\in [1:K]$, when $\bln$ goes large,  for $\al\defeq  [\alpha_1, \alpha_2, \cdots, \alpha_K ]$.
The capacity region $C (P, \al)$ is the collection of  all the achievable  rate tuples $(R_1 (P, \al), R_2 (P, \al), R_c (P, \al))$. 
The GDoF  region $\mathcal{D} ( \al)$ is defined as 
 \begin{align}
\mathcal{D} ( \al)   \defeq & \Big\{ (d_1, d_2, \cdots, d_K) :         \exists  \bigl( R_1 (P, \al), R_2 (P, \al), \cdots, R_K (P, \al) \bigr) \in C(P, \al) \non\\
 & \quad s. t. \quad   d_k = \lim_{P \to \infty}   \frac{  R_k (P, \al)}{ \frac{1}{2} \log P}, \  \forall k \in [1: K]    \Big\} .  \non
  \end{align}
 The   sum GDoF is then defined  by \[ d_{\text{sum}}   (\al) \defeq     \max_{\substack{d_1, d_2, \cdots, d_K:\\ (d_1, d_2, \cdots, d_K) \in \mathcal{D} ( \al)} }    d_1+ d_2 + \cdots + d_K.\]
GDoF is a generalization of the DoF.  Note that DoF can be considered as a specific point of GDoF by letting $\alpha_{1}= \alpha_{2}= \cdots = \alpha_{K}=1$.

\section{Main result  \label{sec:mainresult}}

The main result of this work is the characterization of the optimal sum GDoF for the $K$-user asymmetric interference channel.

\begin{theorem}  \label{thm:GDoFIC}
For the $K$-user asymmetric interference channel defined in Section~\ref{sec:system},  for almost   all  realizations of channel coefficients  $\{h_{k\ell}\}$,  the optimal  sum GDoF  is characterized as 
\begin{align}
\dsum(\al) = \frac{  \sum_{k=1}^K \alpha_k  +  \alpha_K -\alpha_{K-1} }{2}  . \label{thm:GDoF1} 
    \end{align}
\end{theorem}
\begin{proof} 
The achievability is based on  multi-layer interference alignment and successive decoding.
The converse for  this \emph{asymmetric} case involves bounding the \emph{weighted} sum GDoF for selected $J+2$ users, $J \in [1: \lceil \log \frac{K}{2} \rceil ]$.
The details of the achievability and converse proofs are provided in Section~\ref{sec:achievability} and Section~\ref{sec:converse}, respectively. 
\end{proof}

\begin{remark}
The result of Theorem~\ref{thm:GDoFIC}  matches the previous result $\dsum(\alpha) = \frac{K}{2}$ when $\alpha_1 =   \alpha_2 = \cdots =  \alpha_K =1$ (see \cite{CJ:08}). 
\end{remark}
\begin{remark}
One observation from the result  of Theorem~\ref{thm:GDoFIC} is that,  the change of the link strength of the $(K-1)$th receiver, i.e., $\alpha_{K-1}$,  will not take effect on the sum GDoF, as long as $\alpha_{K-2}\leq \alpha_{K-1} \leq \alpha_{K}$.
\end{remark}
\begin{remark}
From the result  of Theorem~\ref{thm:GDoFIC}, it reveals that the link strength of the $K$th receiver, i.e., $\alpha_{K}$, takes more effect on the optimal sum GDoF (with a larger weight), compared to the link strengths of the other receivers.
\end{remark}

\section{Converse   \label{sec:converse}}

This section provides the converse of Theorem~\ref{thm:GDoFIC}, for the  $K$-user \emph{asymmetric} interference channel defined in Section~\ref{sec:system}. 
While the converse for the \emph{symmetric} case only requires bounding the sum DoF for selected \emph{two} users, the converse for  this \emph{asymmetric} case involves bounding the \emph{weighted} sum GDoF for selected $J+2$ users,  with corresponding weights  $(2^{J}, 2^{J-1}, \cdots, 2^{2}, 2^{1}, 1, 1)$, for $J \in [1: \lceil \log \frac{K}{2} \rceil ]$.
The result on  bounding the weighted sum GDoF is given in the following lemma.

\begin{lemma}  \label{lm:dLsum}
For $1\leq l_1< l_2 <  \cdots  < l_{J+2} \leq K$ and $J \in [1: \lceil \log \frac{K}{2} \rceil ]$,  then the following inequality holds true 
\begin{align}
\sum_{j=1}^J  2^{J-j+1} d_{l_j} + d_{l_{J+1}} +  d_{l_{J+2}}   \leq \sum_{j=1}^J 2^{J-j} \alpha_{l_j} +  \alpha_{l_{J+2}} .
 \end{align}
\end{lemma}

Before proving Lemma~\ref{lm:dLsum}, let us provide the following result derived from Lemma~\ref{lm:dLsum}, which serves as the converse of Theorem~\ref{thm:GDoFIC}.

\begin{corollary}  \label{cor:ubound}
For the $K$-user asymmetric interference channel defined in Section~\ref{sec:system}, the optimal  sum GDoF  is upper bounded by 
\begin{align}
\dsum(\al) \leq  \frac{  \sum_{k=1}^K \alpha_k  +  \alpha_K -\alpha_{K-1} }{2}  . \label{eq:ubound11} 
    \end{align}
\end{corollary}
\begin{proof} 
The proof is based on Lemma~\ref{lm:dLsum}. The details of this proof are  provided in Appendix~\ref{app:ubound}. 
\end{proof}

Let us now prove Lemma~\ref{lm:dLsum}.  
Without loss of generality, we will focus on the case of $l_i =i$ for $i \in [1: J+2]$  and $J \in [1: \lceil \log \frac{K}{2} \rceil ]$, and prove 
\begin{align}
\sum_{j=1}^J  2^{J-j+1} d_{j} + d_{J+1} +  d_{J+2}   \leq \sum_{j=1}^J 2^{J-j} \alpha_{j} +  \alpha_{J+2}.   \label{eq:dLsumWLG11} 
 \end{align}
Let us define an auxiliary variable
\begin{align} 
\tilde{y}_{k, \ell}(t) & \defeq   \sqrt{P^{\alpha_{\ell}}} \sum_{i=1}^{K} h_{ki}  x_{i}(t)  +  \tilde{z}_{\ell} (t)   \label{eq:ytildel}   
\end{align} 
where $\tilde{z}_{\ell} (t)   \sim \mathcal{N}(0, 1)$ is independent of the other noise random variables, for $k, \ell \in [1:K]$. 
 Let  $y^{\bln}_k \defeq \{y_k (t)\}_{t=1}^{\bln}$,  $x^{\bln}_k \defeq \{x_k (t)\}_{t=1}^{\bln}$, $z^{\bln}_k \defeq \{z_k (t)\}_{t=1}^{\bln}$, and $\tilde{y}^{\bln}_{k,\ell} \defeq \{\tilde{y}_{k, \ell} (t)\}_{t=1}^{\bln}$. 
For the ease of description, we  define that \[\bar{W}_{[i,j]} \defeq \{  w_{\ell}:  \ell \in [1: K],   \ell \neq i,  \ell \neq j\}\] and $\bar{W}_{[i]} \defeq \{  w_{\ell}:  \ell \in [1: K],   \ell \neq i\}$, for $i, j  \in [1:K], i\neq j$.
We also define that
  \begin{align}
 \Phi(\Jo) \defeq 2^{J-\Jo +1}  \Imu( w_{\Jo}; y^{\bln}_{\Jo})   +  \sum_{j=\Jo+1}^{J+2}  2^{\max\{J-j+1, 0\}}  \Imu( w_{j}; \tilde{y}^{\bln}_{\Jo+1, \Jo}  |  \bar{W}_{[j]} )  \label{eq:phidef100}  
 \end{align}
 for $\Jo\in [1: J-1]$,  and that  
   \begin{align}
  d_{0} \defeq 0 ,  \quad  \alpha_{0} \defeq 0 ,  \quad  \tilde{y}^{\bln}_{1, 0} \defeq \phi,  \quad   \Imu( w_{j}; \tilde{y}^{\bln}_{1, 0}  |  \bar{W}_{[j]} ) \defeq 0, \   \forall j ,  \quad  \Imu( w_{0}; y^{\bln}_{0}) \defeq 0, \quad \text{and} \quad  \Phi(0)   \defeq0.   \label{eq:phidef200}  
 \end{align}

Beginning with Fano's inequality, we have
\begin{align}
&\sum_{j=1}^J  2^{J-j+1}  \bln R_{j} + \bln R_{J+1} +  \bln R_{J+2} - n \epsilon_n   \non\\
 \leq  &\sum_{j=1}^{J-1}  2^{J-j+1}  \Imu( w_{j}; y^{\bln}_j ) +   2  \Imu( w_{J}; y^{\bln}_J )   +  \Imu( w_{J+1}; y^{\bln}_{J+1} ) +    \Imu( w_{J+2}; y^{\bln}_{J+2} )   \label{eq:fano112233}\\
  \leq  &\sum_{j=1}^{J-1}  2^{J-j+1}  \Imu( w_{j}; y^{\bln}_j ) +   \sum_{j=J}^{J+2}  2^{\max\{J-j+1,0\}}  \Imu( w_{j}; \tilde{y}^{\bln}_{J, J -1}  |  \bar{W}_{[j]} )    \non\\&  +  \bigl( (\alpha_{J+2}- \alpha_{J}) +2 (\alpha_{J}- \alpha_{J-1})\bigr)  \frac{\bln}{2}   \log P + \bln  o(\log P)   \label{eq:upb00011} \\
   =  &\sum_{j=1}^{J-2}  2^{J-j+1}  \Imu( w_{j}; y^{\bln}_j ) +   \Phi(J-1)  \non\\&  +  \bigl( (\alpha_{J+2}- \alpha_{J}) +2 (\alpha_{J}- \alpha_{J-1})\bigr)  \frac{\bln}{2}   \log P + \bln  o(\log P)   \label{eq:upb00022} \\
   \leq   &\sum_{j=1}^{J-3}  2^{J-j+1}  \Imu( w_{j}; y^{\bln}_j ) +  \Phi(J- 2)
   \non\\&  +  \bigl( (\alpha_{J+2}- \alpha_{J}) +2 (\alpha_{J}- \alpha_{J-1}) +  2^{2} (\alpha_{J-1}  - \alpha_{J-2}) \bigr)  \frac{\bln}{2}   \log P + \bln  o(\log P)   \label{eq:upb00033} \\
   \leq   &\sum_{j=1}^{J-4}  2^{J-j+1}  \Imu( w_{j}; y^{\bln}_j ) +  \Phi(J- 3)
   \non\\&  +  \bigl((\alpha_{J+2}- \alpha_{J}) +2 (\alpha_{J}- \alpha_{J-1}) +  2^{2} (\alpha_{J-1}  - \alpha_{J-2}) + 2^{3} (\alpha_{J-2}  - \alpha_{J-3}) \bigr)  \frac{\bln}{2}   \log P + \bln  o(\log P)   \label{eq:upb00044} \\
      \vdots&  \non\\
      \leq   &  \bigl( (\alpha_{J+2}- \alpha_{J}) +2 (\alpha_{J}- \alpha_{J-1}) + 2^{2} (\alpha_{J-1}  \!-\! \alpha_{J-2}) + 2^{3} (\alpha_{J-2}  \!-\! \alpha_{J-3}) + \! \cdots \!+  2^{J} (\alpha_{1}  \!-\! \alpha_{0}) \bigr)  \frac{\bln}{2}   \log P \non\\&+ \bln  o(\log P)   \label{eq:upb00055} \\
     = & \bigl( \sum_{j=1}^J 2^{J-j} \alpha_{j} +  \alpha_{J+2} \bigr)  \frac{\bln}{2}   \log P  + \bln  o(\log P)  \label{eq:upb00066}
\end{align}
 where  $ \Phi(\Jo)$ is defined in \eqref{eq:phidef100}, for $\Jo\in [1: J-1]$;
\eqref{eq:fano112233} is from Fano's inequality, and $\epsilon_n \to 0$ as $n \to \infty$; 
 \eqref{eq:upb00011} follows from Lemma~\ref{lm:dffbound11}, which is provided at the end of this section;
 \eqref{eq:upb00022} uses the definition of $ \Phi(\Jo)$;
 \eqref{eq:upb00033}-\eqref{eq:upb00055} follow from the result of Lemma~\ref{lm:boundchain998}, provided at the end of this section.
 By dividing each side of \eqref{eq:upb00066} with  $\frac{\bln}{2}   \log P $ and letting $\bln, P \to \infty$,  it proves the bound in  \eqref{eq:dLsumWLG11}.  By mapping the indexes $i$ with $l_i$, for $i \in [1: J+2]$ and $1\leq l_1< l_2 <  \cdots  < l_{J+2} \leq K$,   it then proves Lemma~\ref{lm:dLsum}.

Note that, in our proof the weights of the  sum GDoF for $J+2$ users are designed specifically as $(2^{J}, 2^{J-1}, \cdots, 2^{2}, 2^{1}, 1, 1)$. With this design, for $\Jo \in [1:J ]$, the $\Jo$th mutual information term  $ \Imu( w_{\Jo}; y^{\bln}_{\Jo} )$ with weight $2^{J-\Jo+1}$  can be bounded with other $2^{J-\Jo+1}$  mutual information terms generated from  User~$(\Jo+1)$ to User~$(J+2)$, i.e.,  $\sum_{j=\Jo+1}^{J+2}  2^{\max\{J-j+1, 0\}}  \Imu( w_{j}; \tilde{y}^{\bln}_{\Jo+1, \Jo}  |  \bar{W}_{[j]} )$. This bounding operation also generates a total of  $2^{J-(\Jo-1)+1}$ mutual information terms  that will be used to bound the $(\Jo-1)$th mutual information term  $ \Imu( w_{\Jo-1}; y^{\bln}_{\Jo-1} )$ with weight $2^{J-(\Jo-1)+1}$. This process repeats until $\Jo=1$. 
Since a weighted mutual information term is  bounded with other weighted mutual information terms and it also generates  new  terms for the next operation, it then forms a ``chain'' on this bounding process.

The lemmas and claims used in our proof are provided below. Their proofs are relegated to Appendix~\ref{app:prooflmclaims}. 
 
\begin{lemma}  \label{lm:boundchain998}
For $ \Phi(\Jo)$ defined in \eqref{eq:phidef100},  $\Jo\in [1: J-1]$, we have the following bound
\begin{align}
 \Phi(\Jo)  +  2^{J-(\Jo-1) +1}  \Imu( w_{\Jo-1}; y^{\bln}_{\Jo -1})   
\leq   2^{J-\Jo +1} (\alpha_{\Jo}  - \alpha_{\Jo -1}) \cdot \frac{\bln}{2}   \log P + \bln  o(\log P)  +  \Phi(\Jo -1)    \non 
 \end{align}
  where $\alpha_{0}, \Imu( w_{0}; y^{\bln}_{0})$, and  $\Phi(0) $ are  defined in \eqref{eq:phidef200}.
\end{lemma}
\begin{proof}
See Appendix~\ref{app:boundchain998}. 
The proof is based on the result of Lemma~\ref{lm:boundchain012}. 
\end{proof}
 
 \begin{lemma}  \label{lm:boundchain012}
For $\Jo\in [1: J-1]$, the following inequality is true
\begin{align}
& 2^{J-\Jo +1}  \Imu( w_{\Jo}; y^{\bln}_{\Jo})   +  \sum_{j=\Jo+1}^{J+2}  2^{\max\{J-j+1, 0\}}  \Imu( w_{j}; \tilde{y}^{\bln}_{\Jo+1, \Jo}  |  \bar{W}_{[j]} )   \non\\
\leq &  2^{J-\Jo +1} (\alpha_{\Jo}  - \alpha_{\Jo -1}) \cdot \frac{\bln}{2}   \log P + \bln  o(\log P)  + \sum_{j=\Jo}^{J+2}  2^{\max\{J-j+1,  0\}}  \Imu( w_{j}; \tilde{y}^{\bln}_{\Jo, \Jo -1}  |  \bar{W}_{[j]} )   \non 
 \end{align}
where  $\alpha_{0}, \tilde{y}^{\bln}_{1, 0}$, and $\Imu( w_{j}; \tilde{y}^{\bln}_{1, 0}  |  \bar{W}_{[j]} )$ are defined in \eqref{eq:phidef200}.
\end{lemma}
\begin{proof}
See Appendix~\ref{app:boundchain012}. The proof uses the result of Lemma~\ref{lm:sum2bound12}.
 \end{proof}

 \begin{lemma}  \label{lm:dffbound11}
The following bound holds true
\begin{align}
& 2  \Imu( w_{J}; y^{\bln}_J )   +  \Imu( w_{J+1}; y^{\bln}_{J+1} ) +    \Imu( w_{J+2}; y^{\bln}_{J+2} )   \non\\
\leq &  2 \Imu( w_{J}; \tilde{y}^{\bln}_{J, J -1}  |  \bar{W}_{[J]} )   +  \Imu( w_{J+1};    \tilde{y}^{\bln}_{J, J -1}  | \bar{W}_{[J+1]} )  +     \Imu( w_{J+2};  \tilde{y}^{\bln}_{J, J -1} | \bar{W}_{[ J+2]} )   \non\\&  +  (\alpha_{J+2}- \alpha_{J} +2 (\alpha_{J}- \alpha_{J-1}) ) \cdot \frac{\bln}{2}   \log P + \bln  o(\log P) . \non 
 \end{align}
\end{lemma}
\begin{proof}
See Appendix~\ref{app:dffbound11}. The proof uses the result of Lemma~\ref{lm:sum2bound12}. 
 \end{proof}

\begin{lemma}  \label{lm:sum2bound12}
For $\ell_1,  \ell_2 , \ell_3 ,  l, i, j \in [1:K]$,  $\ell_1 <  \ell_2 \leq \ell_3 $, $i \neq  j$, then the following bound is true
\begin{align}
          &  \Imu( w_{i}; y^{\bln}_{\ell_2} |  \tilde{y}^{\bln}_{\ell_2,\ell_1} ,  \bar{W}_{[i, j]} )  +    \Imu( w_{j}; \tilde{y}^{\bln}_{l, \ell_3} | \tilde{y}^{\bln}_{\ell_2,\ell_1}, \bar{W}_{[j]} ) \non\\
    \leq & \frac{\bln}{2} \log ( 1+ P^{\alpha_{\ell_2} - \alpha_{\ell_1}} )  +  \frac{\bln}{2} \log  \bigl(1  + P^{\alpha_{\ell_3} -\alpha_{\ell_2} } \frac{ |h_{lj } |^2}{|h_{\ell_2 j }|^2}  \bigr)  . \non
 \end{align}
 When $  \ell_2 , \ell_3 ,  l,  j \in [1:K]$ and $ \ell_2 \leq \ell_3 $,  then we have
 \begin{align}
            \Imu( w_{i}; y^{\bln}_{\ell_2} |  \bar{W}_{[i, j]} )  +    \Imu( w_{j}; \tilde{y}^{\bln}_{l, \ell_3} |  \bar{W}_{[j]} ) 
    \leq  \alpha_{\ell_3} \cdot \frac{\bln}{2} \log P  +  \bln  o(\log P)  . \non
 \end{align}
\end{lemma}
\begin{proof}
See Appendix~\ref{app:sum2bound12}. The proof is based on the result of Claim~\ref{lm:bound1JJ1122} and Claim~\ref{lm:bound1JJ3344}. 
 \end{proof}

 \begin{claim}  \label{lm:bound1JJ1122}
 For $\ell_1,  \ell_2 ,  i, j \in [1:K]$,  $\ell_1 <  \ell_2  $, $i \neq  j$,  it holds true that
  \begin{align}
 \Imu( w_{i}, w_{j}; y^{\bln}_{\ell_2} |  \tilde{y}^{\bln}_{\ell_2,\ell_1} ,  \bar{W}_{[i, j]} ) \leq  \frac{\bln}{2} \log ( 1+ P^{\alpha_{\ell_2} - \alpha_{\ell_1}} ) . \non
\end{align}
 When  $ \ell_2 , i, j \in [1:K]$, $i \neq  j$, then the following inequality is true
   \begin{align}
 \Imu( w_{i}, w_{j}; y^{\bln}_{\ell_2} |   \bar{W}_{[i, j]} ) \leq  \alpha_{\ell_2} \cdot \frac{\bln}{2} \log P  +  \bln  o(\log P) . \non
\end{align}
\end{claim}
\begin{proof}
See Appendix~\ref{app:bound1JJ1122}. 
\end{proof}

 \begin{claim}  \label{lm:bound1JJ3344}
For $\ell_1,  \ell_2 , \ell_3 ,  l,  j \in [1:K]$,  $\ell_1 <  \ell_2 \leq \ell_3 $,  it is true that
  \begin{align}
\Imu( w_{j}; \tilde{y}^{\bln}_{l, \ell_3} | y^{\bln}_{\ell_2}, \tilde{y}^{\bln}_{\ell_2,\ell_1}, \bar{W}_{[j]} ) \leq \frac{\bln}{2} \log  \bigl(1  + P^{\alpha_{\ell_3} -\alpha_{\ell_2} } \frac{ |h_{lj } |^2}{|h_{\ell_2 j }|^2}  \bigr) . \non
\end{align}
When $  \ell_2 , \ell_3 ,  l,  j \in [1:K]$,  $ \ell_2 \leq \ell_3 $, and $ \tilde{y}^{\bln}_{\ell_2,\ell_1} = \phi$,  then the above inequality is also true.
\end{claim}
\begin{proof}
See Appendix~\ref{app:bound1JJ3344}. 
 \end{proof}

 \section{Achievability   \label{sec:achievability}}

This section provides the  achievability  for Theorem~\ref{thm:GDoFIC}. 
The achievability is based on  multi-layer interference alignment, where different interference alignment sub-schemes are designed in different  layers associated with specific power levels. In this scheme, the method of successive decoding is applied at the receivers. 
In the proposed scheme,  pulse  amplitude modulation  (PAM) will be used.

Let us first review the  PAM modulation that will be  used in our  scheme.  If  a random variable $x$ is uniformly drawn from 
the following PAM constellation set
\begin{align}
   \Omega (\xi,  Q)  \defeq   \{ \xi \cdot a :   \    a \in  \Zc  \cap [-Q,   Q]   \}      \label{eq:cons0099}  
 \end{align}
for some $Q \in \Zc^+ $ and $\xi \in \Rc$, then the average power of $x$ is 
\begin{align}
 \E |x|^2  =  \frac{2  \xi^2 }{ 2Q +1}  \sum_{i=1}^{Q} i^2  =  \frac{  \xi^2  Q(Q+1)}{3}.  \label{eq:avpowe2233}    
\end{align} 
The parameter $\xi$ is used to regularize  the  average power of $x$.
The expression in \eqref{eq:avpowe2233} implies that 
\begin{align}
 \E |x|^2  \leq 1/\tau,  \quad \text{for}  \quad \xi \leq  \frac{  1}{  \sqrt{\tau}Q}    \label{eq:avpowe4455}  
\end{align} 
given some $\tau > 1$.  One property for the PAM constellation is that, given some PAM signals $c_1, c_2, \cdots, c_M \in \Omega (\xi,  Q)$,  the sum of them is still a PAM signal such that  
 \begin{align}
  c_1 + c_2 + \cdots + c_M  \in  \Omega (\xi,  M Q)   .   \label{eq:cons2288}  
 \end{align}

In the GDoF analysis of the proposed scheme, we will use the Khintchine-Groshev Theorem for Monomials\footnote{A function $f(\vv)$ is a monomial generated by $\vv=(v_1, v_2,  \cdots, v_N) \in \Rc^{N}$ if this function can be written as $f(\vv) = \prod_{i=1}^{N} v_i^{\beta_i}$, for $\beta_i \in \mathbb{N}, \forall i \in[1:N]$.}, which is stated in the following Theorem, as in \cite{MGMK:14}.

\begin{theorem}[Khintchine-Groshev Theorem for Monomials]  \label{thm:KG}
Let $N\leq M$, $\vv=(v_1, v_2, \cdots, v_N) \in \Rc^N$, and $g_1, g_2, \cdots, g_M$ be distinct monomials generated by $\vv$. Then, for any $\epsilon' >0$ and almost all $\vv$, there exists a positive constant $\kappa$ such that
\begin{align}
\big| \sum_{i=1}^{M} g_i q_i \big| > \frac{\kappa}{\max_i | q_i|^{M-1 +\epsilon'}}
 \end{align}
 holds for all $(q_1, q_2, \cdots, q_M) \neq \boldsymbol{0} \in \Zc^M$.
\end{theorem}

Let us describe the proposed scheme with multi-layer interference alignment and successive decoding, given in the following sub-sections.

\subsection{Multi-layer interference alignment}   \label{sec:MLIA}

The proposed scheme consists of $K$ sub-schemes, with each sub-scheme designed in a specific  layer, i.e.,  at a specific power level.  
For each of the first $K-2$ layers, the design follows from the interference alignment technique \cite{CJ:08, MGMK:14}. Since interference alignment is designed across multiple layers, we call it as multi-layer interference alignment. The last two layers are  dedicated to two users and one user, respectively. Thus, the design of the last two layers is very simple.

\begin{figure}
\centering
\includegraphics[width=9cm]{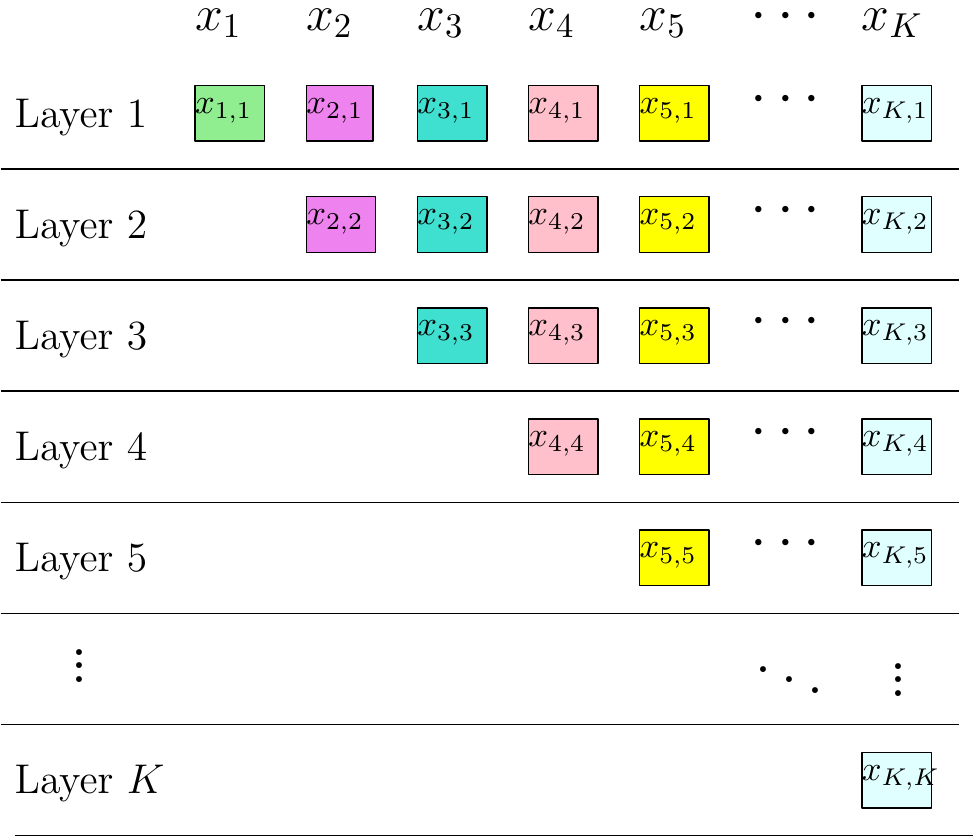}
\caption{The structure of the multi-layer interference alignment. The $\ell$th layer is dedicated to the last $(K-\ell+1)$ users,  from Users $\ell$ to User $K$, for $\ell  \in [1:K]$. For Transmitter~$k$, the transmitted signal is a superposition of the signals dedicated to the first $k$  layers, and $x_{k,\ell} $ is the signal  dedicated  to the $\ell$th layer, for $\ell \in [1:k]$, $k\in [1:K]$.}
\label{fig:MLIA}
\end{figure}

The $\ell$th  layer (the $\ell$th sub-scheme) is  dedicated specifically to the last $\Ke_{\ell}$ users,  from Users $\ell$ to User $K$, where
 \begin{align}
  \Ke_{\ell} \defeq K-\ell+1, \quad  \ell  \in [1:K] . \label{eq:Ke0011}  
\end{align}
For Transmitter~$k$, the transmitted signal is a superposition of the signals dedicated to the first $k$  layers, designed as
\begin{align}
x_{k} &=   \sum_{\ell=1}^{k}  \sqrt{P^{  - \alpha_{\ell-1}}}   x_{k,\ell}  \quad  \text{for} \quad x_{k,\ell} =     \vv_{k,\ell}^\T   \bv_{k,\ell}   \label{eq:xtran9900}  
\end{align}
for $ k\in [1:K]$,  where $ \alpha_{0} \defeq 0$ and $x_{k,\ell} $ is the signal of Transmitter~$k$ dedicated  to the $\ell$th layer.
The vector 
\begin{align}
\vv_{k,\ell} \defeq  [ v_{k,\ell, 1}, v_{k,\ell, 2},  \cdots, v_{k,\ell, \Nd_{\ell}}  ]^\T \in \Rc^{\Nd_{\ell} \times 1}  \label{eq:vv5566}  
\end{align}
 will be specified  later on, where  $\Nd_{\ell}$ is designed as 
 \begin{subnumcases}  
{   \Nd_{\ell} \defeq } 
     \me^{\Ke_{\ell}(\Ke_{\ell}-1) }     &    if    \  $ \ell  \in [1:K-2] $      \label{eq:Nd111} \\
1  &  if  \ $\ell \in [K-1: K]$   \label{eq:Nd222}
\end{subnumcases}
for some $\me \in \Zc^+$.
The vector 
\begin{align}
\bv_{k,\ell} \defeq  [ b_{k,\ell, 1}, b_{k,\ell, 2},  \cdots, b_{k,\ell, \Nd_{\ell}}  ]^\T   \label{eq:bv5566}  
\end{align}
 is an information vector for the $\ell$th layer, where  the elements $\{b_{k,\ell, i}\}_{i=1}^{\Nd_{\ell}}$ are  \emph{independent} random variables   \emph{uniformly} drawn from the following  PAM constellation set\footnote{Without loss of generality we will assume that $P^{ \frac{ \lambda_{\ell} }{2}}$ is an integer, for $\ell\in [1:K]$. 
When $P^{ \frac{ \lambda_{\ell} }{2}}$ isn't an integer,  we can slightly modify the parameter $\epsilon$  in \eqref{eq:lambda111} and \eqref{eq:lambda222} such that $P^{ \frac{ \lambda_{\ell} }{2}}$ is an integer, for the regime with large $P$.} 
 \begin{align}
   b_{k,\ell, i}      &  \in    \Omega ( \xi =  \gamma  \cdot \frac{ 1}{Q_{\ell}} ,   \   Q =  Q_{\ell} ),  \quad   i \in [1:\Nd_{\ell} ],  \  k \in [\ell : K], \ \ell \in [1 : K]   \label{eq:cons0295}  
 \end{align}
where   $\gamma $  is a  positive constant, and $Q_{\ell}$ is defined as
 \begin{align}
  Q_{\ell} \defeq P^{ \frac{ \lambda_{\ell} }{2}}   ,  \quad \ell \in [1 : K] .  \label{eq:lambda000}
 \end{align}
The parameter $\lambda_{\ell}$ is designed as  
 \begin{subnumcases}  
{   \lambda_{\ell} \defeq } 
     \frac{\alpha_{\ell} - \alpha_{\ell-1}}{\Md_{\ell}} - \epsilon     &    if    \  $ \ell  \in [1:K-2] $      \label{eq:lambda111} \\
    \frac{\alpha_{\ell} - \alpha_{\ell-1}}{K-\ell +1}  - \epsilon &  if  \ $\ell \in [K-1: K]$  \label{eq:lambda222} 
\end{subnumcases}
for 
 \begin{align}
  \Md_{\ell}    &  \defeq 2 \me^{\Ke_{\ell}(\Ke_{\ell}-1) } + (\Ke_{\ell}-1) \me^{\Ke_{\ell}(\Ke_{\ell}-1) -1 } -1  \label{eq:defmd000}  
 \end{align}
and for some  small enough  $\epsilon > 0$. 
As we will see later on,  $ \lambda_{\ell}$ represents  the GDoF carried by each of the symbols $\{b_{k,\ell, i}\}_{i, k}$.
In our scheme, when $\alpha_{\ell} = \alpha_{\ell-1}$, then the $\ell$th layer can be simply removed without affecting the GDoF performance, i.e., the signal $x_{k,\ell}$ is set as  $x_{k,\ell} =0, \forall k$.  Without loss of generality, we will focus on the case with $\alpha_{\ell} > \alpha_{\ell-1}, \forall \ell $.

Let us now design the vectors of $\vv_{k,\ell}$ for each layer. 
The design of $\vv_{k,\ell}$  for the last two layers is very straightforward.   Note that the $(K-1)$th layer is dedicated to User~$K-1$ and User~$K$, while the $K$th layer is dedicated to User~$K$ only. Therefore, we set the parameters as \[v_{K-1,K-1, 1}= v_{K,K-1, 1} =v_{K,K, 1} =1.\] Recall that $\Nd_{K-1} = \Nd_{K} =1$ (see \eqref{eq:Nd222}). In the following, we will design the vectors of $\vv_{k,\ell}$ for the $\ell$th layer, for $\ell \in [1:K-2]$.  
For the $\ell$th layer dedicated to the last $\Ke_{\ell}$ users, we define a set of \emph{dimensions} as
 \begin{align}
   \mathcal{V}_{\ell, \me}      &  \defeq \Bigl\{  \prod_{j=\ell}^{K}  \prod_{\substack{i=\ell\\ i\neq j}}^{K} h_{ij}^{\beta_{ij}} :  \beta_{ij} \in [0: \me -1 ]    \Bigr\}, \quad  \ell \in [1:K-2] .  \label{eq:defvelln}  
 \end{align}
 Note that $\mathcal{V}_{\ell, \me}$ consists of  $ \Nd_{\ell}$ rationally independent real numbers\footnote{We say $p_1, p_2, \cdots, p_M$ are rationally independent if  the only $M$-tuple of integers $q_1, q_2, \cdots, q_M$ such that $\sum_{i=1}^M p_iq_i =0$ is the trivial solution in which every $q_i$ is zero.}, where $\Nd_{\ell} = \me^{\Ke_{\ell}(\Ke_{\ell}-1)}$   for $\ell \in [1:K-2] $.  
In our scheme, we let $\vv_{k,\ell}$ be the vector containing all the elements in set $\mathcal{V}_{\ell, \me}$, i.e.,
 \begin{align}
    v_{k,\ell, i} =  \mathcal{V}_{\ell, \me} (i), \quad    i \in [1:\Nd_{\ell} ],  k \in [\ell : K], \ell \in [1 : K-2] .   \label{eq:vkli000}  
 \end{align}
$\mathcal{V}_{\ell, \me}  (i)$ denotes the $i$th element of the set $\mathcal{V}_{\ell, \me}$.

Based on our design, Lemma~\ref{lm:powerans} (see below) shows that the average power of each transmitted signal is upper bounded by $\gamma^2    \eta$, where $\eta$ is a positive value independent of $P$, and  $\gamma $  is a  positive constant appeared in \eqref{eq:cons0295}. 
Thus, by setting $\gamma$ as a constant that is bounded away from zero and is no more than $\frac{1}{\sqrt{\eta}}$, i.e., $\gamma \in (0, \frac{1}{\sqrt{\eta}}]$, then the  average power  constraint is satisfied, that is, $\E|x_{k}|^2 \leq 1$ for $k\in [1:K]$.

\begin{lemma}  \label{lm:powerans}
Based on the signal design in \eqref{eq:xtran9900}-\eqref{eq:defvelln},  the average power of the transmitted signal at  Transmitter~$k$, $k\in [1:K]$,  satisfies
\begin{align}
\E|x_{k}|^2 & \leq  \gamma^2    \eta   \label{eq:power6265}  
\end{align}
where  $\eta$ is a positive value independent of $P$. 
\end{lemma}
\begin{proof}
See Appendix~\ref{app:powerans}. 
\end{proof}

\subsection{Successive decoding}  \label{sec:SD}

The decoding is based on successive decoding. The idea of successive decoding is to decode the signals for one layer by treating the lower layers as noise,  and then remove them to decode the signals in the next layer. The signals decoded in one layer include the desired signals and the interference signals that might be in a certain form.

Let us first focus on the decoding for the first $K-2$ layers, and then discuss the decoding for the last two layers. 
For the $\ell$th layer, $\ell \in [1:K-2] $, based on the above design of multi-layer interference alignment,  at  Receiver~$k$, $k\in [\ell:K]$, the interference signals  can be aligned into a set of  \emph{dimensions} denoted by $\Ic_{k,\ell}$, for 
 \begin{align}
 \Ic_{k,\ell} =  \bigcup_{\substack{l\in [\ell:K]\\ l\neq k}} \Bigl\{ h_{kl}^{m} \cdot  \prod_{\substack{i, j \in [\ell:K]\\ i\neq j\\ (i,j)\neq (k,l) } }   h_{ij}^{\beta_{ij}} :  \beta_{ij} \in [0: \me -1 ]    \Bigr\} \bigcup \Bigl\{ \mathcal{V}_{\ell, \me }  \big\backslash \bigl\{1\bigr\}  \Bigr\}        \label{eq:defvellnm11}  
 \end{align}
which satisfies $ \Ic_{k,\ell} \subset  \mathcal{V}_{\ell, \me +1}$ and \[|\Ic_{k,\ell}| = \me^{\Ke_{\ell}(\Ke_{\ell}-1) } + (\Ke_{\ell}-1) \me^{\Ke_{\ell}(\Ke_{\ell}-1) -1 } -1 = \Md_{\ell} - \Nd_{\ell};\] 
while the desired signals lie in a set of  dimensions denoted by $\mathcal{S}_{k,\ell}$,  for
 \begin{align}
 \mathcal{S}_{k,\ell} =   h_{kk} \mathcal{V}_{\ell, \me}      &  = \Bigl\{  h_{kk} \prod_{j=\ell}^{K}  \prod_{\substack{i=\ell\\ i\neq j}}^{K} h_{ij}^{\beta_{ij}} :  \beta_{ij} \in [0: \me -1 ]    \Bigr\} \label{eq:defvelln0099}  
 \end{align}
which satisfies \[|\mathcal{S}_{k,\ell}|= \me^{\Ke_{\ell}(\Ke_{\ell}-1) } = \Nd_{\ell}.\]  Note that $h_{kk}$ is not appeared in the dimensions of $\Ic_{k,\ell}$. Also note that $h_{kk}$ is  appeared in each dimension of $\mathcal{S}_{k,\ell}$. It then implies that all the dimensions in $\Ic_{k,\ell} \cup  \mathcal{S}_{k,\ell}$ are rationally independent.

For the successive decoding at  the $\ell$th layer, $\ell \in [1:K-2]$, at Receiver~$k$, $k\in [\ell:K]$,  the goal  is to decode the desired information vector $\bv_{k,\ell}$ (see \eqref{eq:bv5566}), as well as the interference at that layer, given that the decoding of the previous layers is complete. 
For  the $\ell$th layer, $\ell \in [1:K-2]$, assuming that the decoding  of the previous layers is complete, then  Receiver~$k$, $k\in [\ell:K]$ has the following observation (removing the time index)
\begin{align}
y_{k,\ell} &\defeq y_{k} - \underbrace{ \sum_{l=1}^{\ell-1}  \sum_{j=l}^{K}   \sqrt{P^{  \alpha_{k}- \alpha_{l-1}}}  h_{kj} \vv_{j,l}^\T   \bv_{j,l}}_{\text{side information  from  previous layers}}    \label{eq:ykl22331} 
\end{align}
where the term of $\sum_{l=1}^{\ell-1}  \sum_{j=l}^{K}  \sqrt{P^{  \alpha_{k}- \alpha_{l-1}}}  h_{kj} \vv_{j,l}^\T   \bv_{j,l}$  is constructed from  the side information about desired signals and interference obtained from the decoding of the previous layers, with $ \sum_{l=1}^{0} s_i  \defeq 0 $ for any  $s_i \in  \Rc$.  When $\ell=1$,  this term is zero.
Let us expand $y_{k,\ell}$ from \eqref{eq:ykl22331} to the following expression:
\begin{align}
y_{k,\ell} &=  \sum_{l=1}^{K}  \sum_{j=l}^{K}  \sqrt{P^{  \alpha_{k}- \alpha_{l-1}}}  h_{kj} \vv_{j,l}^\T   \bv_{j,l}  +z_{k}  -  \sum_{l=1}^{\ell-1}  \sum_{j=l}^{K} \sqrt{P^{  \alpha_{k}- \alpha_{l-1}}}  h_{kj} \vv_{j,l}^\T   \bv_{j,l}   \non\\   
&=  \underbrace{\sqrt{P^{  \alpha_{k} - \alpha_{\ell-1}}} h_{kk}  \vv_{k,\ell}^\T   \bv_{k,\ell}}_{\defeq S_{k,\ell}, \  \text{desired signal}}  +  \underbrace{\sum_{\substack{j=\ell \\ j\neq k}}^{K}   \sqrt{P^{  \alpha_{k} - \alpha_{\ell-1}}} h_{kj}  \vv_{j,\ell}^\T   \bv_{j,\ell}}_{\defeq I_{k,\ell}, \ \text{interference}}  +    \underbrace{ \sum_{l=\ell+1}^{K}  \sum_{j=l}^{K}  \sqrt{P^{  \alpha_{k}- \alpha_{l-1}}}   h_{kj} \vv_{j,l}^\T   \bv_{j,l} }_{\defeq T_{k,\ell}, \  \text{treated as noise}} +z_{k}  \label{eq:channelGK0123}  
\end{align}
where   
\begin{align}
S_{k,\ell} \defeq \sqrt{P^{  \alpha_{k} - \alpha_{\ell-1}}} h_{kk}  \vv_{k,\ell}^\T   \bv_{k,\ell}, \quad  I_{k,\ell} \defeq \sum_{\substack{j=\ell \\ j\neq k}}^{K}   \sqrt{P^{  \alpha_{k} - \alpha_{\ell-1}}} h_{kj}  \vv_{j,\ell}^\T   \bv_{j,\ell},  \quad T_{k,\ell} \defeq  \sum_{l=\ell+1}^{K}  \sum_{j=l}^{K}  \sqrt{P^{  \alpha_{k}- \alpha_{l-1}}}   h_{kj} \vv_{j,l}^\T   \bv_{j,l} \label{eq:SITdef000}  
\end{align}
for   $k\in [\ell:K]$,  $\ell \in [1:K-2]$. 
From the above expression, $y_{k,\ell}$ can be expanded into four terms: $S_{k,\ell}$, $I_{k,\ell}$,  $T_{k,\ell}$ and noise.  For Receiver~$k$, $S_{k,\ell}$  corresponds to the term containing  desired information at Layer~$\ell$;  $I_{k,\ell}$ represents the  interference  at  Layer~$\ell$;  and $T_{k,\ell}$ denotes the term containing signals dedicated to the next layers, which can be treated as noise. 
The term $S_{k,\ell}$ can be rewritten in the following form 
\begin{align}
S_{k,\ell}  =  \gamma \sqrt{P^{ \alpha_{k} - \alpha_{\ell-1} - \lambda_{\ell}}}  \sum_{i=1}^{|\mathcal{S}_{k,\ell}|}   \mathcal{S}_{k,\ell} (i) q_{k,\ell,i}  \quad \text{for}  \quad q_{k,\ell,1},  \cdots, q_{k,\ell, |\mathcal{S}_{k,\ell}|} \in [-Q_{\ell}: Q_{\ell} ]  \label{eq:channelGK2345}  
\end{align}
where $Q_{\ell}$ and $\lambda_{\ell}$ are defined in \eqref{eq:lambda000}, \eqref{eq:lambda111} and \eqref{eq:lambda222}. From \eqref{eq:cons0295} it holds true that $q_{k,\ell,i} \defeq b_{k,\ell, i} \cdot  \frac{P^{ \frac{ \lambda_{\ell} }{2}}}{\gamma}  \in [- Q_{\ell}, Q_{\ell}]$, for $i \in [1:\Nd_{\ell} ]$,  $k \in [\ell : K]$, $\ell \in [1 : K-2]$. 
Similarly, the interference term $I_{k,\ell}$ can be expressed in the form of  
\begin{align}
I_{k,\ell}  =  \gamma \sqrt{P^{ \alpha_{k} - \alpha_{\ell-1} - \lambda_{\ell}}}  \sum_{i=1}^{| \Ic_{k,\ell}|}  \Ic_{k,\ell} (i) q'_{k,\ell,i}  \quad \text{for}  \quad q'_{k,\ell,1},  \cdots, q'_{k,\ell, |\Ic_{k,\ell}|} \in [-\Ke_{\ell}Q_{\ell}: \Ke_{\ell}Q_{\ell} ]  \label{eq:channelGK3456}  
\end{align}
Note that, if the PAM signals lie at the same dimension,  the sum of PAM signals is still a PAM signal. 
In the above expression, $q'_{k,\ell,i}$ represents the sum of the normalized PAM signals (normalized by $\gamma  P^{ - \frac{ \lambda_{\ell} }{2}} $) lying at the dimension  $\Ic_{k,\ell} (i)$, and thus $q'_{k,\ell,i} \in [-\Ke_{\ell}Q_{\ell}: \Ke_{\ell}Q_{\ell} ] $  for $i \in [1: | \Ic_{k,\ell}| ]$,  $k \in [\ell : K]$, $\ell \in [1 : K-2]$. 
In this layer, the goal is to decode $q_{k,\ell,1},  \cdots, q_{k,\ell, |\mathcal{S}_{k,\ell}|}, q'_{k,\ell,1},  \cdots, q'_{k,\ell, |\Ic_{k,\ell}|}$ from $y_{k,\ell}$ by treating $T_{k,\ell}$ as noise.

Let us now focus on the minimum distance of the  constellation for the signal  $S_{k,\ell} +I_{k,\ell}$,  which is defined by
  \begin{align}
 d_{\min}(k, \ell)   \defeq        \!\!\!\!   \!\!\!\! \!\!\!\!  \min_{\substack{ q_{k,\ell,1},  \cdots, q_{k,\ell, |\mathcal{S}_{k,\ell}|}, q'_{k,\ell,1},  \cdots, q'_{k,\ell, |\Ic_{k,\ell}|} :  \\  q_{k,\ell,1},  \cdots, q_{k,\ell, |\mathcal{S}_{k,\ell}|} \in [-Q_{\ell}: Q_{\ell} ]  \\ q'_{k,\ell,1},  \cdots, q'_{k,\ell, |\Ic_{k,\ell}|} \in [-\Ke_{\ell}Q_{\ell}: \Ke_{\ell}Q_{\ell} ]   \\  (q_{k,\ell,1},  \cdots, q_{k,\ell, |\mathcal{S}_{k,\ell}|}, q'_{k,\ell,1},  \cdots, q'_{k,\ell, |\Ic_{k,\ell}|}) \neq  (0, 0 , \cdots, 0)  }}         \!\!\!\! \!\!\!\!  \gamma \sqrt{P^{ \alpha_{k} - \alpha_{\ell-1} - \lambda_{\ell}}} \ \Big|  \sum_{i=1}^{|\mathcal{S}_{k,\ell}|}   \mathcal{S}_{k,\ell} (i) q_{k,\ell,i}  +   \sum_{i=1}^{| \Ic_{k,\ell}|}  \Ic_{k,\ell} (i) q'_{k,\ell,i} \Big|   \label{eq:mindis7788}
 \end{align}
 for $k \in [\ell : K]$, $\ell \in [1 : K-2]$. 
 For the minimum distance $ d_{\min}(k, \ell) $ defined in  \eqref{eq:mindis7788}, Lemma~\ref{lm:mindiskl} (shown at the end of this section) provides  a result on its lower bound.
 On the other hand, for the term $T_{k,\ell}$ appeared in \eqref{eq:channelGK0123}, Lemma~\ref{lm:boundTkl} (shown at the end of this section) provides  a result on its upper bound.
 Let us go back to the expression of $y_{k,\ell}$ (see \eqref{eq:channelGK0123}), that is,
 \begin{align}
y_{k,\ell} &=  S_{k,\ell} +I_{k,\ell} + T_{k,\ell} +  z_{k}  \label{eq:ykl9922}  
\end{align}
  for $k\in [\ell:K]$, $\ell \in [1:K-2]$. From Lemma~\ref{lm:boundTkl},  $T_{k,\ell}$  is upper bounded by 
  $ T_{k,\ell} \leq    P^{  \frac{\alpha_{k}- \alpha_{\ell}}{2}} \cdot \delta_{k,\ell}$, where $\delta_{k,\ell}$ is a positive value independent of $P$. 
  From Lemma~\ref{lm:mindiskl},  the  minimum distance  of the  constellation for the signal  $S_{k,\ell} +I_{k,\ell}$  is lower bounded by 
  $d_{\min}(k, \ell)   \geq   \kappa'  P^{   \frac{ \alpha_{k} - \alpha_{\ell}   + \epsilon_{\ell} }{2} }$, for any small enough $\epsilon_{\ell}>0$, where $\kappa'$ is a positive constant. 
 Therefore,  one can easily show that $q_{k,\ell,1},  \cdots, q_{k,\ell, |\mathcal{S}_{k,\ell}|}, q'_{k,\ell,1},  \cdots, q'_{k,\ell, |\Ic_{k,\ell}|}$ can be decoded from $y_{k,\ell}$ by treating $T_{k,\ell}$ as noise, with vanishing error probability as $P$ goes large. 
 At this point, at Layer~$\ell$, the information vector $\bv_{k,\ell}$ is decoded at Receiver~$k$, and the interference $I_{k,\ell}$ can be reconstructed by Receiver~$k$ with the side information of $q'_{k,\ell,1},  \cdots, q'_{k,\ell, |\Ic_{k,\ell}|}$, for  $k\in [\ell:K]$, $\ell \in [1:K-2]$.

Once the decoding at  Layer~$\ell$ is complete,  Receiver~$k$ removes the reconstructed $S_{k,\ell}$ and $I_{k,\ell}$ from $y_{k,\ell}$, and then moves onto the decoding at the next layer, i.e.,  Layer~$(\ell +1)$, for $k\in [\ell+1:K]$, $\ell +1  \in [2:K-2]$. 
 
The decoding at the last two layers is very straightforward.   Note that the $(K-1)$th layer is dedicated to User~$K-1$ and User~$K$, while the $K$th layer is dedicated to User~$K$ only. Recall that, $\Nd_{K-1} = \Nd_{K} =1$, $v_{K-1,K-1, 1}= v_{K,K-1, 1} =v_{K,K, 1} =1$, and
 \begin{align}
 x_{K-1,K-1} &=   b_{K-1,K-1, 1} \in \Omega ( \xi =  \gamma  \cdot \frac{ 1}{Q_{K-1}} ,   \   Q =  Q_{K-1} ) \non\\
 x_{K,K-1} &=  b_{K,K-1, 1}  \in \Omega ( \xi =  \gamma  \cdot \frac{ 1}{Q_{K-1}} ,   \   Q =  Q_{K-1} )   \non\\
 x_{K,K} &=   b_{K,K,1}  \in \Omega ( \xi =  \gamma  \cdot \frac{ 1}{Q_{K}} ,   \   Q =  Q_{K} )   \non
 \end{align}
for $Q_{K-1} \defeq P^{ \frac{ (\alpha_{K-1} - \alpha_{K-2})/2  - \epsilon}{2}}$ and $Q_{K} \defeq P^{ \frac{ \alpha_{K} - \alpha_{K-1}  - \epsilon}{2}}$. 
Once the decoding of the first $K-2$ layers is complete, both Receiver~$(K-1)$ and Receiver~$K$ remove all the intended signals and interference signals dedicated to the first $K-2$ layers from the corresponding received observations.  
After that, for the $(K-1)$th layer, the decoding problem  is simply equivalent to decoding two symbols at  a $2\times2$ interference channel with sum GDoF $\alpha_{K-1} - \alpha_{K-2}$, where the SNR of this channel is $P^{ \alpha_{K-1} - \alpha_{K-2}}$. 
One can easily show that this two symbols can be decoded at both Receiver~$(K-1)$ and Receiver~$K$ with vanishing error probability as $P$ goes large. After that, Receiver~$K$ removes the decoded symbols and then decodes its only one symbol at the last layer. At this point, the whole decoding is complete.  

After successive decoding for all the layers,   Receiver~$k$,  $k \in [1 : K]$, is able to decode all the following PAM symbols
 \begin{align}
   b_{k,\ell, i}      &  \in    \Omega ( \xi =  \gamma  \cdot \frac{ 1}{P^{ \frac{ \lambda_{\ell} }{2}}} ,   \   Q =  P^{ \frac{ \lambda_{\ell} }{2}}),  \quad \forall  i \in [1:\Nd_{\ell} ],  \ \ell \in [1 : k]   \label{eq:rate8877}  
 \end{align}
where   $\lambda_{\ell} $ is defined in \eqref{eq:lambda111} and \eqref{eq:lambda222}.  Since  $b_{k,\ell, i}$ is independently and uniformly  drawn from the corresponding   PAM constellation $ \Omega ( \xi =  \gamma  \cdot \frac{ 1}{P^{ \frac{ \lambda_{\ell} }{2}}} ,   \   Q =  P^{ \frac{ \lambda_{\ell} }{2}})$, then $b_{k,\ell, i}$ carries the following amount of  bits of information
 \begin{align}
  \Hen(b_{k,\ell, i}) &  =  \log  (1+ 2 P^{ \frac{ \lambda_{\ell} }{2}})  = \frac{ \lambda_{\ell} }{2}  \log  P  + o(\log P) \label{eq:rate2255}  
 \end{align}
 for  $i \in [1:\Nd_{\ell} ]$, $ \ell \in [1 : k]$, $k \in [1 : K]$.
By summing up all the amount of information carried by all the symbols from  all the users, and considering that those symbols are sent over a single channel use, it implies that for almost all  realizations of channel coefficients the following  sum rate is achievable when $P$ is large
 \begin{align}
R_{sum} = &  \sum_{k=1}^{K}  R_k   \non\\
= &  \sum_{k=1}^{K}  \sum_{\ell=1}^{k}  \sum_{i=1}^{\Nd_{\ell}}  \Hen(b_{k,\ell, i})    \non\\
 =  &  \sum_{k=1}^{K}  \sum_{\ell=1}^{k}  \sum_{i=1}^{\Nd_{\ell}} \bigl( \frac{ \lambda_{\ell} }{2}  \log  P   + o(\log P) \bigr)  \label{eq:rate225566}  \\
    =  &  \sum_{\ell=1}^{K}  \sum_{k=\ell}^{K}  \frac{ \Nd_{\ell} \lambda_{\ell} }{2}  \log  P + o(\log P)  \non\\
        =  &  \sum_{\ell=1}^{K}   \frac{ \Nd_{\ell}    \lambda_{\ell} ( K-\ell +1) }{2}  \log  P  + o(\log P) \non\\
      =  & \sum_{\ell=1}^{K-2}   \frac{ \Nd_{\ell} ( K-\ell +1)   ( \frac{\alpha_{\ell} - \alpha_{\ell-1}}{\Md_{\ell}} - \epsilon  ) }{2}  \log  P   \non\\& +         \frac{ 2(\frac{\alpha_{K-1} - \alpha_{K-2}}{2}  - \epsilon) }{2}  \log  P  + \frac{ \alpha_{K} - \alpha_{K-1} - \epsilon  }{2}  \log  P   + o(\log P) \label{eq:rate3435}
 \end{align}
 where \eqref{eq:rate225566} follows  from \eqref{eq:rate2255}. 
Recall that $\lambda_{\ell} = \frac{\alpha_{\ell} - \alpha_{\ell-1}}{\Md_{\ell}} - \epsilon$   if  $ \ell  \in [1:K-2]$, and $\lambda_{\ell} = \frac{\alpha_{\ell} - \alpha_{\ell-1}}{K-\ell +1}  - \epsilon$   if  $ \ell \in [K-1: K]$. 
For the sum rate expressed in \eqref{eq:rate3435}, by dividing each side with  $\frac{1}{2}  \log  P$ and letting $P \to \infty$ and  $\epsilon \to 0$,  it reveals that for almost all  realizations of channel coefficients the following sum GDoF is achievable
 \begin{align}
\dsum^{achievable}(\al) = &   \sum_{\ell=1}^{K-2}    ( K-\ell +1)  (\alpha_{\ell} - \alpha_{\ell-1})  \frac{ \Nd_{\ell}}{\Md_{\ell}}    +   \frac{   2(\alpha_{K-1} - \alpha_{K-2})}{2}    +  \alpha_{K} - \alpha_{K-1}.   \label{eq:rate2743} 
 \end{align}
Note that when $ \ell  \in [1:K-2]$, we have $\frac{ \Nd_{\ell}}{\Md_{\ell}} = \frac{\me^{\Ke_{\ell}(\Ke_{\ell}-1) }}{2 \me^{\Ke_{\ell}(\Ke_{\ell}-1) } + (\Ke_{\ell}-1) \me^{\Ke_{\ell}(\Ke_{\ell}-1) -1 } -1}$, which converges to $\frac{1}{2}$ for large enough $\me$. 
Therefore, for large enough $\me$, the achievable sum GDoF expressed in \eqref{eq:rate2743} can be simplified as 
 \begin{align}
\dsum^{achievable}(\al) = &   \sum_{\ell=1}^{K-2}   \frac{ ( K-\ell +1)  (\alpha_{\ell} - \alpha_{\ell-1})}{ 2}    +      \frac{2(\alpha_{K-1} - \alpha_{K-2})}{2}    +  \alpha_{K} - \alpha_{K-1}  \non\\
= &   \frac{  \sum_{k=1}^K \alpha_k  +  \alpha_K -\alpha_{K-1} }{2}     \label{eq:rateGDoFfinal} 
 \end{align}
which holds for almost all  realizations of channel coefficients.  At this point, we complete the achievability proof for Theorem~\ref{thm:GDoFIC}.
 The two lemmas used in the GDoF analysis are provided below. 
 
 \begin{lemma}  \label{lm:mindiskl}
Consider the  minimum distance  $d_{\min}(k, \ell)$ defined in \eqref{eq:mindis7788}.   For almost all  realizations of channel coefficients, and for any small enough $\epsilon_{\ell}>0$, there exists a positive constant $\kappa'$ such that
  \begin{align}
 d_{\min}(k, \ell)   \geq   \kappa'  P^{   \frac{ \alpha_{k} - \alpha_{\ell}   + \epsilon_{\ell} }{2} }   \non
 \end{align}
 for $k \in [\ell : K]$, $\ell \in [1 : K-2]$. 
\end{lemma}
\begin{proof}
See Appendix~\ref{app:mindiskl}. The proof uses the result of  Khintchine-Groshev Theorem for Monomials. 
\end{proof}

 \begin{lemma}  \label{lm:boundTkl}
For the term $T_{k,\ell}$ defined in \eqref{eq:SITdef000},  it can be upper bounded by 
  \begin{align}
T_{k,\ell} \leq  &  P^{  \frac{\alpha_{k}- \alpha_{\ell}}{2}} \cdot \delta_{k,\ell}  \non
  \end{align}
 where $\delta_{k,\ell}$ is a positive value independent of $P$,  for $k \in [\ell : K]$, $\ell \in [1 : K-2]$.
  \end{lemma}
\begin{proof}
See Appendix~\ref{app:boundTkl}. 
\end{proof}

\section{Conclusion}   \label{sec:conclusion}

This work considered  the $K$-user asymmetric  interference channel, where different receivers might have different channel gains, parameterized by $0< \alpha_1 \leq   \alpha_2 \leq \cdots \leq  \alpha_K \leq1$. 
For this channel, we characterized the optimal sum GDoF as $\dsum = \frac{ \sum_{k=1}^K \alpha_k  +  \alpha_K -\alpha_{K-1}}{2}$. 
The achievability is based on  multi-layer interference alignment and  successive decoding.
For the the converse of this asymmetric setting, it involves bounding the \emph{weighted} sum GDoF for selected $J+2$ users, $J \in [1: \lceil \log \frac{K}{2} \rceil]$, which is very different from the case of the symmetric setting that only requires bounding the sum DoF for selected \emph{two} users.
The result of this work generalizes the existing result of the symmetric case to the setting with diverse link strengths.

\appendices

\section{Proofs of Lemmas~\ref{lm:boundchain998}, \ref{lm:boundchain012}, \ref{lm:dffbound11} and \ref{lm:sum2bound12},   and Claims~\ref{lm:bound1JJ1122} and~\ref{lm:bound1JJ3344}} \label{app:prooflmclaims}

Recall that 
  \begin{align}
  \tilde{y}_{k, \ell}(t) & \defeq   \sqrt{P^{\alpha_{\ell}}} \sum_{i=1}^{K} h_{ki}  x_{i}(t)  +  \tilde{z}_{\ell} (t) \non\\
 \Phi(\Jo) &\defeq 2^{J-\Jo +1}  \Imu( w_{\Jo}; y^{\bln}_{\Jo})   +  \sum_{j=\Jo+1}^{J+2}  2^{\max\{J-j+1, 0\}}  \Imu( w_{j}; \tilde{y}^{\bln}_{\Jo+1, \Jo}  |  \bar{W}_{[j]} )  \non
 \end{align} 
   \begin{align}
  d_{0} \defeq 0 ,  \quad  \alpha_{0} \defeq 0 ,  \quad  \tilde{y}^{\bln}_{1, 0} \defeq \phi,  \quad   \Imu( w_{j}; \tilde{y}^{\bln}_{1, 0}  |  \bar{W}_{[j]} ) \defeq 0, \   \forall j ,  \quad  \Imu( w_{0}; y^{\bln}_{0}) \defeq 0, \quad \text{and} \quad  \Phi(0)   \defeq0   \non
 \end{align}
 for   $\Jo\in [1: J-1]$ and $J \in [1: \lceil \log \frac{K}{2} \rceil ]$  (see \eqref{eq:ytildel}, \eqref{eq:phidef100} and \eqref{eq:phidef200}).

\subsection{Proof of Lemma~\ref{lm:boundchain998}  } \label{app:boundchain998}

The proof is based on the result of Lemma~\ref{lm:boundchain012}. 
Specifically,  Lemma~\ref{lm:boundchain012} reveals that
\begin{align}
 \Phi(\Jo)   
\leq   2^{J-\Jo +1} (\alpha_{\Jo}  - \alpha_{\Jo -1}) \cdot \frac{\bln}{2}   \log P + \bln  o(\log P)   + \sum_{j=\Jo}^{J+2}  2^{\max\{J-j+1,  0\}}  \Imu( w_{j}; \tilde{y}^{\bln}_{\Jo, \Jo -1}  |  \bar{W}_{[j]} )    \non 
 \end{align}
  for $\Jo\in [1: J-1]$.  By adding $2^{J-(\Jo-1) +1}  \Imu( w_{\Jo-1}; y^{\bln}_{\Jo -1}) $ into both sides of the above inequality, we have
\begin{align}
 \Phi(\Jo)  +  2^{J-(\Jo-1) +1}  \Imu( w_{\Jo-1}; y^{\bln}_{\Jo -1})  
\leq   2^{J-\Jo +1} (\alpha_{\Jo}  - \alpha_{\Jo -1}) \cdot \frac{\bln}{2}   \log P + \bln  o(\log P)  +  \Phi(\Jo -1)  \non 
 \end{align}
 which completes the proof of Lemma~\ref{lm:boundchain998} .

\subsection{Proof of Lemma~\ref{lm:boundchain012}  } \label{app:boundchain012}

The proof will use the result of Lemma~\ref{lm:sum2bound12}. 
At first, we note that the following equality is true
\begin{align}
& 2^{J-\Jo +1}  \Imu( w_{\Jo}; y^{\bln}_{\Jo})   +  \sum_{j=\Jo+1}^{J+2}  2^{\max\{J-j+1, 0\}}  \Imu( w_{j}; \tilde{y}^{\bln}_{\Jo+1, \Jo}  |  \bar{W}_{[j]} )   \non\\
= &     \sum_{j=\Jo+1}^{J+2}  2^{\max\{J-j+1, 0\}}  \Bigl(  \Imu( w_{\Jo}; y^{\bln}_{\Jo})+ \Imu( w_{j}; \tilde{y}^{\bln}_{\Jo+1, \Jo}  |  \bar{W}_{[j]} ) \Bigr)    \label{eq:ublast2t3100}      
 \end{align}
 by using the identity of $\sum_{j=\Jo+1}^{J+2}  2^{\max\{J-j+1, 0\}}  =2^{J-\Jo +1}$, for  $\Jo\in [1: J-1]$.
 For the sum of two mutual information terms in the right-hand side of \eqref{eq:ublast2t3100}, given $j \in [\Jo+1, J+2]$, we have 
\begin{align}
& \Imu( w_{\Jo}; y^{\bln}_{\Jo})+ \Imu( w_{j}; \tilde{y}^{\bln}_{\Jo+1, \Jo}  |  \bar{W}_{[j]} )   \non\\
\leq  &      \Imu( w_{\Jo}; y^{\bln}_{\Jo},  \tilde{y}^{\bln}_{\Jo, \Jo-1},  \bar{W}_{[j, \Jo]} )+ \Imu( w_{j}; \tilde{y}^{\bln}_{\Jo+1, \Jo} , \tilde{y}^{\bln}_{\Jo, \Jo-1} |  \bar{W}_{[j]}   )   \label{eq:ublast2t3000}\\    
=  &      \underbrace{\Imu( w_{\Jo};  \tilde{y}^{\bln}_{\Jo, \Jo-1} | \bar{W}_{[j, \Jo]} )}_{\leq \Imu( w_{\Jo};  \tilde{y}^{\bln}_{\Jo, \Jo-1} | \bar{W}_{[ \Jo]} )}+ \Imu( w_{j}; \tilde{y}^{\bln}_{\Jo, \Jo-1} |  \bar{W}_{[j]}   ) \non\\&
+\underbrace{\Imu( w_{\Jo}; y^{\bln}_{\Jo} |  \tilde{y}^{\bln}_{\Jo, \Jo-1},  \bar{W}_{[j, \Jo]} )+ \Imu( w_{j}; \tilde{y}^{\bln}_{\Jo+1, \Jo} | \tilde{y}^{\bln}_{\Jo, \Jo-1} , \bar{W}_{[j]}   )}_{ \leq (\alpha_{\Jo}- \alpha_{\Jo-1}) \cdot \frac{\bln}{2}   \log P + \bln  o(\log P) } \label{eq:ublast2t3001}\\  
\leq   &       \Imu( w_{\Jo};  \tilde{y}^{\bln}_{\Jo, \Jo-1} | \bar{W}_{[ \Jo]} )+ \Imu( w_{j}; \tilde{y}^{\bln}_{\Jo, \Jo-1} |  \bar{W}_{[j]}   ) 
\non\\&
+ (\alpha_{\Jo}- \alpha_{\Jo-1}) \cdot \frac{\bln}{2}   \log P + \bln  o(\log P)  \label{eq:ublast2t3002}    
\end{align}
 where  the step in \eqref{eq:ublast2t3000}  follows from the fact that adding more information does not reduce the mutual information;
    the step in \eqref{eq:ublast2t3001}  uses chain rule and  the  fact that the messages are mutually independent;
 the step in \eqref{eq:ublast2t3002} follows from the derivation of $\Imu( w_{\Jo};  \tilde{y}^{\bln}_{\Jo, \Jo-1} | \bar{W}_{[j, \Jo]} ) \leq \Imu( w_{\Jo};  \tilde{y}^{\bln}_{\Jo, \Jo-1},w_j | \bar{W}_{[j, \Jo]} )  = \Imu( w_{\Jo};  \tilde{y}^{\bln}_{\Jo, \Jo-1} | \bar{W}_{[ \Jo]} )$ and from the result of Lemma~\ref{lm:sum2bound12}, which reveals that $\Imu( w_{\Jo}; y^{\bln}_{\Jo} |  \tilde{y}^{\bln}_{\Jo, \Jo-1},  \bar{W}_{[j, \Jo]} )+ \Imu( w_{j}; \tilde{y}^{\bln}_{\Jo+1, \Jo} | \tilde{y}^{\bln}_{\Jo, \Jo-1} , \bar{W}_{[j]}   ) \leq (\alpha_{\Jo}- \alpha_{\Jo-1}) \cdot \frac{\bln}{2}   \log P + \bln  o(\log P) $.
 
 By incorporating the result of \eqref{eq:ublast2t3002} into \eqref{eq:ublast2t3100}, it gives 
 \begin{align}
& 2^{J-\Jo +1}  \Imu( w_{\Jo}; y^{\bln}_{\Jo})   +  \sum_{j=\Jo+1}^{J+2}  2^{\max\{J-j+1, 0\}}  \Imu( w_{j}; \tilde{y}^{\bln}_{\Jo+1, \Jo}  |  \bar{W}_{[j]} )   \non\\
\leq  &     \sum_{j=\Jo+1}^{J+2}  2^{\max\{J-j+1, 0\}}  \Bigl(  \Imu( w_{\Jo};  \tilde{y}^{\bln}_{\Jo, \Jo-1} | \bar{W}_{[ \Jo]} )+ \Imu( w_{j}; \tilde{y}^{\bln}_{\Jo, \Jo-1} |  \bar{W}_{[j]}   )  
+ (\alpha_{\Jo}- \alpha_{\Jo-1})  \frac{\bln}{2}   \log P + \bln  o(\log P)  \Bigr)    \label{eq:ublast2t3103}      \\
=&2^{J-\Jo +1} (\alpha_{\Jo}  - \alpha_{\Jo -1}) \cdot \frac{\bln}{2}   \log P + \bln  o(\log P)  + \sum_{j=\Jo}^{J+2}  2^{\max\{J-j+1,  0\}}  \Imu( w_{j}; \tilde{y}^{\bln}_{\Jo, \Jo -1}  |  \bar{W}_{[j]} )  \label{eq:ublast2t3104} 
 \end{align}
 where \eqref{eq:ublast2t3103} is from \eqref{eq:ublast2t3002} and \eqref{eq:ublast2t3100};
 \eqref{eq:ublast2t3104} follows from the identity of $\sum_{j=\Jo+1}^{J+2}  2^{\max\{J-j+1, 0\}}  =2^{J-\Jo +1}$, for $\Jo\in [1: J-1]$.
 Then, we complete the proof of Lemma~\ref{lm:boundchain012}.

 \subsection{Proof of Lemma~\ref{lm:dffbound11}  } \label{app:dffbound11}
 
The proof will use the result of  Lemma~\ref{lm:sum2bound12}. 
 In the first step, we expand $2  \Imu( w_{J}; y^{\bln}_J ) $ as follows
\begin{align}
 2  \Imu( w_{J}; y^{\bln}_J )    
\leq  &    \Imu( w_{J}; y^{\bln}_J, \tilde{y}^{\bln}_{J, J -1}, \bar{W}_{[J,J+1]} )   +      \Imu( w_{J}; y^{\bln}_J , \tilde{y}^{\bln}_{J, J -1}, \bar{W}_{[J,J+2]})  \label{eq:ublast2t2111} \\ 
=  & \Imu( w_{J}; \tilde{y}^{\bln}_{J, J -1}  |  \bar{W}_{[J,J+1]} )   +      \Imu( w_{J};  \tilde{y}^{\bln}_{J, J -1}  | \bar{W}_{[J,J+2]})   \non\\&  +   \Imu( w_{J}; y^{\bln}_J |  \tilde{y}^{\bln}_{J, J -1} , \bar{W}_{[J,J+1]} ) +  \Imu( w_{J}; y^{\bln}_J |  \tilde{y}^{\bln}_{J, J -1} , \bar{W}_{[J,J+2]})  \label{eq:ublast2t2222} \\ 
\leq   & \Imu( w_{J}; \tilde{y}^{\bln}_{J, J -1}  |  \bar{W}_{[J]} )   +      \Imu( w_{J};  \tilde{y}^{\bln}_{J, J -1}  | \bar{W}_{[J]})   \non\\&  +   \Imu( w_{J}; y^{\bln}_J |  \tilde{y}^{\bln}_{J, J -1} , \bar{W}_{[J,J+1]} ) +  \Imu( w_{J}; y^{\bln}_J |  \tilde{y}^{\bln}_{J, J -1} , \bar{W}_{[J,J+2]})  \label{eq:ublast2t2333}  
 \end{align}
  where   \eqref{eq:ublast2t2111}  follows from the fact that adding more information does not reduce the mutual information;
     \eqref{eq:ublast2t2222}  uses chain rule and  the  fact that the messages are mutually independent;
  and \eqref{eq:ublast2t2333} results from the derivation that $ \Imu( w_{J}; \tilde{y}^{\bln}_{J, J -1}  |  \bar{W}_{[J,\ell]} ) \leq \Imu( w_{J}; \tilde{y}^{\bln}_{J, J -1}, w_{\ell}  |   \bar{W}_{[J,\ell]}) = \Imu( w_{J}; \tilde{y}^{\bln}_{J, J -1}  |  \bar{W}_{[J]} )$ for $\ell\in [1:K]$, $\ell \neq J$.

In the second step,  we expand $\Imu( w_{J+1}; y^{\bln}_{J+1} ) +    \Imu( w_{J+2}; y^{\bln}_{J+2} ) $ as follows
\begin{align}
  &\Imu( w_{J+1}; y^{\bln}_{J+1} ) +    \Imu( w_{J+2}; y^{\bln}_{J+2} )    \non\\
  \leq  &  \Imu( w_{J+1}; y^{\bln}_{J+1}, \tilde{y}^{\bln}_{J+1,J}, \bar{W}_{[J+1, J+2]} ) +    \Imu( w_{J+2}; y^{\bln}_{J+2}, \tilde{y}^{\bln}_{J+1,J}, \bar{W}_{[J+2]} )    \label{eq:ublast2t111} \\
    =  & \Imu( w_{J+1};  \tilde{y}^{\bln}_{J+1,J}  | \bar{W}_{[J+1, J+2]} )  +    \Imu( w_{J+2};  \tilde{y}^{\bln}_{J+1,J} | \bar{W}_{[J+2]} )   \non\\
    &  + \Imu( w_{J+1}; y^{\bln}_{J+1} |  \tilde{y}^{\bln}_{J+1,J} ,  \bar{W}_{[J+1, J+2]} ) +    \Imu( w_{J+2}; y^{\bln}_{J+2} | \tilde{y}^{\bln}_{J+1,J}, \bar{W}_{[J+2]} )    \label{eq:ublast2t222}  \\
 \leq   & \Imu( w_{J+1};  \tilde{y}^{\bln}_{J+1,J} ,   \tilde{y}^{\bln}_{J, J -1} , w_{J+2} | \bar{W}_{[J+1, J+2]} )  +    \Imu( w_{J+2};  \tilde{y}^{\bln}_{J+1,J}, \tilde{y}^{\bln}_{J, J -1} | \bar{W}_{[J+2]} )   \non\\
    &  + \Imu( w_{J+1}; y^{\bln}_{J+1} |  \tilde{y}^{\bln}_{J+1,J} ,  \bar{W}_{[J+1, J+2]} ) +    \Imu( w_{J+2}; y^{\bln}_{J+2} | \tilde{y}^{\bln}_{J+1,J}, \bar{W}_{[J+2]} )     \label{eq:ublast2t444} \\
=    &     \Imu( w_{J+1};    \tilde{y}^{\bln}_{J, J -1}  | \bar{W}_{[J+1]} )  +     \Imu( w_{J+2};  \tilde{y}^{\bln}_{J, J -1} | \bar{W}_{[ J+2]} ) \non\\
  & + \Imu( w_{J+1};  \tilde{y}^{\bln}_{J+1,J} |  \tilde{y}^{\bln}_{J, J -1},  \bar{W}_{[J+1]} )   +    \Imu( w_{J+2};  \tilde{y}^{\bln}_{J+1,J} |  \tilde{y}^{\bln}_{J, J -1}, \bar{W}_{[J+2]} )  \non\\
    & + \underbrace{ \Imu( w_{J+1}; y^{\bln}_{J+1} |  \tilde{y}^{\bln}_{J+1,J} ,  \bar{W}_{[J+1, J+2]} ) +    \Imu( w_{J+2}; y^{\bln}_{J+2} | \tilde{y}^{\bln}_{J+1,J}, \bar{W}_{[J+2]} ) }_{\leq (\alpha_{J+2}- \alpha_{J}) \cdot \frac{\bln}{2}   \log P + \bln  o(\log P) }    \label{eq:ublast2t333}  \\
\leq    & (\alpha_{J+2}- \alpha_{J}) \cdot \frac{\bln}{2}   \log P + \bln  o(\log P)  + \Imu( w_{J+1};    \tilde{y}^{\bln}_{J, J -1}  | \bar{W}_{[J+1]} )  +     \Imu( w_{J+2};  \tilde{y}^{\bln}_{J, J -1} | \bar{W}_{[ J+2]} )  \non\\
&+ \Imu( w_{J+1};  \tilde{y}^{\bln}_{J+1,J} |  \tilde{y}^{\bln}_{J, J -1},  \bar{W}_{[J+1]} )   +    \Imu( w_{J+2};  \tilde{y}^{\bln}_{J+1,J} |  \tilde{y}^{\bln}_{J, J -1}, \bar{W}_{[J+2]} )    \label{eq:ublast2t666} 
 \end{align}
 where   \eqref{eq:ublast2t111} and  \eqref{eq:ublast2t444}  result from the fact that adding more information does not reduce the mutual information;
  \eqref{eq:ublast2t222} and \eqref{eq:ublast2t333} use  chain rule and the fact that the messages are mutually independent;
   \eqref{eq:ublast2t666} follows from the result of  Lemma~\ref{lm:sum2bound12}, that is, $ \Imu( w_{J+1}; y^{\bln}_{J+1} |  \tilde{y}^{\bln}_{J+1,J} ,  \bar{W}_{[J+1, J+2]} ) +    \Imu( w_{J+2}; y^{\bln}_{J+2} | \tilde{y}^{\bln}_{J+1,J}, \bar{W}_{[J+2]} )  = \Imu( w_{J+1}; y^{\bln}_{J+1} |  \tilde{y}^{\bln}_{J+1,J} ,  \bar{W}_{[J+1, J+2]} ) +   \Imu( w_{J+2}; \tilde{y}^{\bln}_{J+2,J+2} | \tilde{y}^{\bln}_{J+1,J}, \bar{W}_{[J+2]} ) \leq (\alpha_{J+2}- \alpha_{J}) \cdot \frac{\bln}{2}   \log P + \bln  o(\log P)$.

 By combining the results of \eqref{eq:ublast2t2333} and \eqref{eq:ublast2t666}, we have
 \begin{align}
& 2  \Imu( w_{J}; y^{\bln}_J )   +  \Imu( w_{J+1}; y^{\bln}_{J+1} ) +    \Imu( w_{J+2}; y^{\bln}_{J+2} )   \non\\
\leq & 2 \Imu( w_{J}; \tilde{y}^{\bln}_{J, J -1}  |  \bar{W}_{[J]} )   +  \Imu( w_{J+1};    \tilde{y}^{\bln}_{J, J -1}  | \bar{W}_{[J+1]} )  +     \Imu( w_{J+2};  \tilde{y}^{\bln}_{J, J -1} | \bar{W}_{[ J+2]} )     \non\\&  +   \underbrace{\Imu( w_{J}; y^{\bln}_J |  \tilde{y}^{\bln}_{J, J -1} , \bar{W}_{[J,J+1]} )    + \Imu( w_{J+1};  \tilde{y}^{\bln}_{J+1,J} |  \tilde{y}^{\bln}_{J, J -1},  \bar{W}_{[J+1]} )}_{\leq (\alpha_{J}- \alpha_{J-1}) \cdot \frac{\bln}{2}   \log P + \bln  o(\log P) }  \non\\
&+ \underbrace{ \Imu( w_{J}; y^{\bln}_J |  \tilde{y}^{\bln}_{J, J -1} , \bar{W}_{[J,J+2]})  +    \Imu( w_{J+2};  \tilde{y}^{\bln}_{J+1,J} |  \tilde{y}^{\bln}_{J, J -1}, \bar{W}_{[J+2]} ) }_{\leq (\alpha_{J}- \alpha_{J-1}) \cdot \frac{\bln}{2}   \log P + \bln  o(\log P) }  \non  \\
&+ (\alpha_{J+2}- \alpha_{J}) \cdot \frac{\bln}{2}   \log P + \bln  o(\log P) \label{eq:ublast2t18424}   \\
\leq & 2 \Imu( w_{J}; \tilde{y}^{\bln}_{J, J -1}  |  \bar{W}_{[J]} )   +  \Imu( w_{J+1};    \tilde{y}^{\bln}_{J, J -1}  | \bar{W}_{[J+1]} )  +     \Imu( w_{J+2};  \tilde{y}^{\bln}_{J, J -1} | \bar{W}_{[ J+2]} )     
\non\\&  +   (\alpha_{J}- \alpha_{J-1}) \cdot \frac{\bln}{2}   \log P + \bln  o(\log P)   \non\\
&+ (\alpha_{J}- \alpha_{J-1}) \cdot \frac{\bln}{2}   \log P + \bln  o(\log P)   \non  \\
&+ (\alpha_{J+2}- \alpha_{J}) \cdot \frac{\bln}{2}   \log P + \bln  o(\log P) \label{eq:ublast2t233566}  
 \end{align}  
 where \eqref{eq:ublast2t18424} is from \eqref{eq:ublast2t2333} and \eqref{eq:ublast2t666};
 \eqref{eq:ublast2t233566} follows from Lemma~\ref{lm:sum2bound12}.  At this point, we complete the proof of Lemma~\ref{lm:dffbound11}.

  \subsection{Proof of Lemma~\ref{lm:sum2bound12}  } \label{app:sum2bound12}

The proof will use the result of Claim~\ref{lm:bound1JJ1122} and Claim~\ref{lm:bound1JJ3344}. 
When $\ell_1,  \ell_2 , \ell_3 ,  l, i, j \in [1:K]$,  $\ell_1 <  \ell_2 \leq \ell_3 $, $i \neq  j$, we have
  \begin{align}
            &  \Imu( w_{i}; y^{\bln}_{\ell_2} |  \tilde{y}^{\bln}_{\ell_2,\ell_1} ,  \bar{W}_{[i, j]} )  +    \Imu( w_{j}; \tilde{y}^{\bln}_{l, \ell_3} | \tilde{y}^{\bln}_{\ell_2,\ell_1}, \bar{W}_{[j]} ) \non\\
  \leq &  \Imu( w_{i}; y^{\bln}_{\ell_2} |  \tilde{y}^{\bln}_{\ell_2,\ell_1} ,  \bar{W}_{[i, j]} )  +    \Imu( w_{j}; \tilde{y}^{\bln}_{l, \ell_3}, y^{\bln}_{\ell_2} | \tilde{y}^{\bln}_{\ell_2,\ell_1}, \bar{W}_{[j]} )   \label{eq:ublast2t24355} \\ 
    = &  \Imu( w_{i}; y^{\bln}_{\ell_2} |  \tilde{y}^{\bln}_{\ell_2,\ell_1} ,  \bar{W}_{[i, j]} )  +    \Imu( w_{j};  y^{\bln}_{\ell_2} | \tilde{y}^{\bln}_{\ell_2,\ell_1}, \bar{W}_{[j]} )+    \Imu( w_{j}; \tilde{y}^{\bln}_{l, \ell_3} | y^{\bln}_{\ell_2}, \tilde{y}^{\bln}_{\ell_2,\ell_1}, \bar{W}_{[j]} )  \non\\ 
      = &  \underbrace{\Imu( w_{i}, w_{j}; y^{\bln}_{\ell_2} |  \tilde{y}^{\bln}_{\ell_2,\ell_1} ,  \bar{W}_{[i, j]} )}_{\leq  \frac{\bln}{2} \log ( 1+ P^{\alpha_{\ell_2} - \alpha_{\ell_1}} )} +   \underbrace{ \Imu( w_{j}; \tilde{y}^{\bln}_{l, \ell_3} | y^{\bln}_{\ell_2}, \tilde{y}^{\bln}_{\ell_2,\ell_1}, \bar{W}_{[j]} ) }_{\leq \frac{\bln}{2} \log  \bigl(1  + P^{\alpha_{\ell_3} -\alpha_{\ell_2} } \frac{ |h_{lj } |^2}{|h_{\ell_2 j }|^2}  \bigr)}  \non\\ 
      \leq  &  \frac{\bln}{2} \log ( 1+ P^{\alpha_{\ell_2} - \alpha_{\ell_1}} ) +   \frac{\bln}{2} \log  \bigl(1  + P^{\alpha_{\ell_3} -\alpha_{\ell_2} } \frac{ |h_{lj } |^2}{|h_{\ell_2 j }|^2}  \bigr)      \label{eq:ublast2t99880} 
  \end{align}     
  where \eqref{eq:ublast2t24355} uses  the  fact that adding  information does not reduce the mutual information;
and \eqref{eq:ublast2t99880} follows from Claim~\ref{lm:bound1JJ1122} and Claim~\ref{lm:bound1JJ3344}.

 Similarly, when $  \ell_2 , \ell_3 ,  l,  j \in [1:K]$ and $ \ell_2 \leq \ell_3 $,  we have
  \begin{align}
            &  \Imu( w_{i}; y^{\bln}_{\ell_2} |   \bar{W}_{[i, j]} )  +    \Imu( w_{j}; \tilde{y}^{\bln}_{l, \ell_3} |  \bar{W}_{[j]} ) \non\\
  \leq &  \Imu( w_{i}; y^{\bln}_{\ell_2} |    \bar{W}_{[i, j]} )  +    \Imu( w_{j}; \tilde{y}^{\bln}_{l, \ell_3}, y^{\bln}_{\ell_2} |  \bar{W}_{[j]} )   \non \\ 
    = &  \Imu( w_{i}; y^{\bln}_{\ell_2} |    \bar{W}_{[i, j]} )  +    \Imu( w_{j};  y^{\bln}_{\ell_2} | \bar{W}_{[j]} )+    \Imu( w_{j}; \tilde{y}^{\bln}_{l, \ell_3} | y^{\bln}_{\ell_2},  \bar{W}_{[j]} )  \non\\ 
      = &  \underbrace{\Imu( w_{i}, w_{j}; y^{\bln}_{\ell_2} |   \bar{W}_{[i, j]} )}_{\leq  \alpha_{\ell_2} \cdot \frac{\bln}{2} \log P  +  \bln  o(\log P)} +   \underbrace{ \Imu( w_{j}; \tilde{y}^{\bln}_{l, \ell_3} | y^{\bln}_{\ell_2}, \bar{W}_{[j]} ) }_{\leq \frac{\bln}{2} \log  \bigl(1  + P^{\alpha_{\ell_3} -\alpha_{\ell_2} } \frac{ |h_{lj } |^2}{|h_{\ell_2 j }|^2}  \bigr)}  \non\\ 
      \leq  &  \alpha_{\ell_2} \cdot \frac{\bln}{2} \log P  +  \bln  o(\log P) +   \frac{\bln}{2} \log  \bigl(1  + P^{\alpha_{\ell_3} -\alpha_{\ell_2} } \frac{ |h_{lj } |^2}{|h_{\ell_2 j }|^2}  \bigr)      \label{eq:ublast2t8847} \\
            =  &  \alpha_{\ell_3} \cdot \frac{\bln}{2} \log P  +  \bln  o(\log P)  \non
  \end{align}     
  where  \eqref{eq:ublast2t8847} follows from Claim~\ref{lm:bound1JJ1122} and Claim~\ref{lm:bound1JJ3344}.
Then,  we complete the proof of  Lemma~\ref{lm:sum2bound12}.

 \subsection{Proof of Claim~\ref{lm:bound1JJ1122}  } \label{app:bound1JJ1122}
         
When  $\ell_1,  \ell_2 , i, j \in [1:K]$,  $\ell_1 <  \ell_2  $, $i \neq  j$,  we have
  \begin{align}
  &\Imu( w_{i}, w_{j}; y^{\bln}_{\ell_2} |  \tilde{y}^{\bln}_{\ell_2,\ell_1} ,  \bar{W}_{[i, j]} )  \non\\
  =&\hen(y^{\bln}_{\ell_2} |  \tilde{y}^{\bln}_{\ell_2,\ell_1} ,  \bar{W}_{[i, j]}) - \hen(y^{\bln}_{\ell_2} |  \tilde{y}^{\bln}_{\ell_2,\ell_1} ,  \bar{W}_{[i, j]}, w_{i},w_{j} ) \non\\ 
    =&\hen(y^{\bln}_{\ell_2} |  \tilde{y}^{\bln}_{\ell_2,\ell_1} ,  \bar{W}_{[i, j]}) -  \hen(z^{\bln}_{\ell_2} ) \non\\
      =&\hen( \{y_{\ell_2}(t)  -  \sqrt{P^{\alpha_{\ell_2} - \alpha_{\ell_1}}}  \tilde{y}_{\ell_2,\ell_1}  (t)  \}_{t=1}^{\bln} |  \tilde{y}^{\bln}_{\ell_2,\ell_1} ,  \bar{W}_{[i, j]}) -  \hen(z^{\bln}_{\ell_2} ) \non\\
    =&\hen( \{z_{\ell_2}(t)  -  \sqrt{P^{\alpha_{\ell_2} - \alpha_{\ell_1}}}  \tilde{z}_{\ell_1}  (t)  \}_{t=1}^{\bln} |  \tilde{y}^{\bln}_{\ell_2,\ell_1} ,  \bar{W}_{[i, j]}) -  \hen(z^{\bln}_{\ell_2} ) \non\\
     \leq &\hen( \{z_{\ell_2}(t)  -  \sqrt{P^{\alpha_{\ell_2} - \alpha_{\ell_1}}}  \tilde{z}_{\ell_1}  (t)  \}_{t=1}^{\bln}) -  \hen(z^{\bln}_{\ell_2} ) \label{eq:ublast2t23141} \\
        = & \frac{\bln}{2} \log (2\pi e ( 1+ P^{\alpha_{\ell_2} - \alpha_{\ell_1}} ))  -  \frac{\bln}{2} \log (2\pi e)    \non\\  
         = & \frac{\bln}{2} \log ( 1+ P^{\alpha_{\ell_2} - \alpha_{\ell_1}} )   \non
\end{align}
where \eqref{eq:ublast2t23141} follows from the fact that conditioning reduces differential entropy. 

When  $ \ell_2 , i, j \in [1:K]$, $i \neq  j$, we have
  \begin{align}
    &\Imu( w_{i}, w_{j}; y^{\bln}_{\ell_2} |   \bar{W}_{[i, j]} )  \non\\
    =&\hen(y^{\bln}_{\ell_2} |   \bar{W}_{[i, j]}) -  \hen(z^{\bln}_{\ell_2} ) \non\\
      =& \sum_{t=1}^{\bln}\hen( y_{\ell_2}(t)   | y_{\ell_2}^{t-1}  ,  \bar{W}_{[i, j]}) - \frac{\bln}{2} \log (2\pi e) \non\\
\leq & \sum_{t=1}^{\bln}\hen( y_{\ell_2}(t) ) -  \hen(z^{\bln}_{\ell_2} ) \non\\
     \leq &  \frac{\bln}{2} \log (2\pi e ( 1+ P^{\alpha_{\ell_2}} \sum_{k=1}^K | h_{\ell_2 k} |^2  ))  -  \frac{\bln}{2} \log (2\pi e)  \label{eq:ublast2t2356} \\
         = & \alpha_{\ell_2} \cdot \frac{\bln}{2} \log P  +  \bln  o(\log P)  \non
\end{align}
where \eqref{eq:ublast2t2356} uses  the fact that Gaussian input maximizes the differential entropy.  It then completes the proof of Claim~\ref{lm:bound1JJ1122}.

  \subsection{Proof of Claim~\ref{lm:bound1JJ3344}  } \label{app:bound1JJ3344}
  
When $\ell_1,  \ell_2 , \ell_3 ,  l,  j \in [1:K]$, $\ell_1 <  \ell_2 \leq \ell_3 $, or when $  \ell_2 , \ell_3 ,  l,  j \in [1:K]$,  $ \ell_2 \leq \ell_3 $, $ \tilde{y}^{\bln}_{\ell_2,\ell_1} = \phi$,   we have
  \begin{align}
        &  \Imu( w_{j}; \tilde{y}^{\bln}_{l, \ell_3} | y^{\bln}_{\ell_2}, \tilde{y}^{\bln}_{\ell_2,\ell_1}, \bar{W}_{[j]} ) \non\\
        =& \hen( \tilde{y}^{\bln}_{l, \ell_3} | y^{\bln}_{\ell_2}, \tilde{y}^{\bln}_{\ell_2,\ell_1}, \bar{W}_{[j]} ) -\hen( \tilde{y}^{\bln}_{l, \ell_3} | y^{\bln}_{\ell_2}, \tilde{y}^{\bln}_{\ell_2,\ell_1}, \bar{W}_{[j]}, w_{j} )  \non\\
                =& \hen \Bigl( \bigl\{  \sqrt{P^{\alpha_{\ell_3}}}  h_{lj }  x_{j}(t) + \tilde{z}_{\ell_3} (t)   \bigr\}_{t=1}^{\bln}  \ \big|   \bigl\{  \sqrt{P^{\alpha_{\ell_2}}}  h_{\ell_2 j }  x_{j}(t) + z_{\ell_2} (t)   \bigr\}_{t=1}^{\bln}     , \tilde{y}^{\bln}_{\ell_2,\ell_1}, \bar{W}_{[j]} \Bigr) -\hen( \tilde{z}^{\bln}_{\ell_3} )  \non\\
              =& \hen \Bigl( \bigl\{  \sqrt{P^{\alpha_{\ell_3}}}  h_{l j }  x_{j}(t) + \tilde{z}_{\ell_3} (t)   -  \sqrt{P^{\alpha_{\ell_3}- \alpha_{\ell_2}}} \frac{h_{lj }}{h_{\ell_2 j }} \bigl( \sqrt{P^{\alpha_{\ell_2}}}  h_{\ell_2  j }  x_{j}(t) + z_{\ell_2} (t)   \bigr) \bigr\}_{t=1}^{\bln} \  \big|  \non\\ &\quad  \bigl\{  \sqrt{P^{\alpha_{\ell_2}}}  h_{\ell_2 j }  x_{j}(t) + z_{\ell_2} (t)   \bigr\}_{t=1}^{\bln}     , \tilde{y}^{\bln}_{\ell_2,\ell_1}, \bar{W}_{[j]} \Bigr) -\hen( \tilde{z}^{\bln}_{\ell_3} )    \non\\
               \leq & \hen \Bigl( \bigl\{   \tilde{z}_{\ell_3} (t)   -  \sqrt{P^{\alpha_{\ell_3}- \alpha_{\ell_2}}} \frac{h_{lj }}{h_{\ell_2 j }} z_{\ell_2} (t)   \bigr\}_{t=1}^{\bln} \Bigr) -\hen( \tilde{z}^{\bln}_{\ell_3} )   \label{eq:ublast2t23451}\\
            = & \frac{\bln}{2} \log  \bigl(1  + P^{\alpha_{\ell_3}- \alpha_{\ell_2}} \frac{|h_{lj }|^2}{|h_{\ell_2 j }|^2}  \bigr)  \non
 \end{align}
 where \eqref{eq:ublast2t23451} follows from the fact that conditioning reduces differential entropy.   It then completes the proof of Claim~\ref{lm:bound1JJ3344}.

\section{Proof of Corollary~\ref{cor:ubound}} \label{app:ubound}

We will  first prove Corollary~\ref{cor:ubound} for  some specific cases in order to get some insights. After that, we will prove Corollary~\ref{cor:ubound} for the general case.
The proof is based on the result of Lemma~\ref{lm:dLsum}.
At first we define that  $\Jl \defeq  \lceil \log \frac{K}{2} \rceil$ and that 
\begin{subnumcases}  
{\Theta(x)  \defeq } 
      x      &    if    \  $  x  \geq    2^{\Jl}   $      \label{eq:thetadef111} \\
0  &  else    . \label{eq:thetadef222}
\end{subnumcases}
Recall that  (see \eqref{eq:phidef200})
\begin{align}
d_{0}  \defeq 0,  \quad  \alpha_{0} \defeq 0.    \label{eq:thetadef333}
 \end{align}
In our proof, a total of $2^{\Jl}$  bounds are required. 
Among those $2^{\Jl}$  bounds,  the first $2^{\Jl-1}$  bounds have a specific structure. The last $2^{\Jl-1}$  bounds have a similar structure but some elements with certain indexes are erased (set as zeros).

\subsection{Proof for the case with $K=8$} 
From Lemma~\ref{lm:dLsum}, the following bounds hold true
\begin{align}
4 d_1 + 2 d_3 + d_7 +  d_8 & \leq  2\alpha_1 + \alpha_3 + \alpha_8  \non  \\
4 d_2 + 2 d_3 + d_7 +  d_8 & \leq  2\alpha_2 + \alpha_3 + \alpha_8  \non  \\
4 d_4 +  2 d_6 + d_7 +  d_8 & \leq 2\alpha_4 + \alpha_6 + \alpha_8  \non  \\
4 d_5 + 2 d_6 + d_7 +  d_8 & \leq  2\alpha_5 +  \alpha_6 + \alpha_8  . \non  
\end{align}

By summing up the above $4$ bounds and dividing each side with $4$, it gives  $\dsum(\al) \leq  \frac{  \sum_{k=1}^8 \alpha_k  +  \alpha_8 -\alpha_{7} }{2}$.

\subsection{Proof for the case with $K=9$} 
The result of  Lemma~\ref{lm:dLsum} reveals that 
\begin{align}
8 d_1 + 4 d_5 +  2 d_7 + d_8 +  d_9 & \leq  4\alpha_1 + 2\alpha_5 + \alpha_7+ \alpha_9  \non  \\
8 d_2 + 4 d_5 +  2 d_7+ d_8 +  d_9 & \leq  4\alpha_2 + 2\alpha_5 + \alpha_7+ \alpha_9 \non  \\
8 d_3 +  4 d_6 +  2 d_7+ d_8 +  d_9 & \leq 4\alpha_3 + 2\alpha_6 + \alpha_7+ \alpha_9  \non  \\
8 d_4 + 4 d_6 +  2 d_7+ d_8 +  d_9& \leq  4\alpha_4 +  2\alpha_6 + \alpha_7+ \alpha_9  \non   \\
                d_8 +  d_9 & \leq  \    \ \quad  \quad \quad \quad \quad  \quad \quad  \alpha_9  \non  \\
             d_8 +  d_9 & \leq  \    \ \quad  \quad \quad \quad \quad  \quad \quad  \alpha_9  \non  \\
             d_8 +  d_9 & \leq  \    \ \quad  \quad \quad \quad \quad  \quad \quad  \alpha_9  \non  \\
               d_8 +  d_9 & \leq  \    \ \quad  \quad \quad \quad \quad  \quad \quad  \alpha_9.  \non  
\end{align}

By summing up the above $8$ bounds and dividing each side with $8$, we have  $\dsum(\al) \leq  \frac{  \sum_{k=1}^9 \alpha_k  +  \alpha_9 -\alpha_{8} }{2}$.

\subsection{Proof for the case with $K=10$} 
The following bounds are directly derived from Lemma~\ref{lm:dLsum}
\begin{align}
8 d_1 + 4 d_5 +  2 d_7 + d_9 +  d_{10} & \leq  4\alpha_1 + 2\alpha_5 + \alpha_7+ \alpha_{10}  \non  \\
8 d_2 + 4 d_5 +  2 d_7+ d_9 +  d_{10} & \leq  4\alpha_2 + 2\alpha_5 + \alpha_7+ \alpha_{10} \non  \\
8 d_3 +  4 d_6 +  2 d_7+ d_9 +  d_{10} & \leq 4\alpha_3 + 2\alpha_6 + \alpha_7+ \alpha_{10}  \non  \\
8 d_4 + 4 d_6 +  2 d_7+ d_9 +  d_{10} & \leq  4\alpha_4 +  2\alpha_6 + \alpha_7+ \alpha_{10}  \non   \\
   2 d_8 +  d_9 +  d_{10} & \leq  \     \quad  \quad \quad \quad \quad    \alpha_8+ \alpha_{10}  \non  \\
    2 d_8 + d_9 + d_{10} & \leq \     \quad  \quad \quad \quad \quad    \alpha_8+ \alpha_{10}  \non  \\
      2 d_8 +    d_9 +  d_{10} & \leq  \     \quad  \quad \quad \quad \quad    \alpha_8+ \alpha_{10}  \non  \\
        2d_8 +    d_9 +  d_{10} & \leq  \     \quad  \quad \quad \quad \quad    \alpha_8+ \alpha_{10}.  \non  
\end{align}

By combining the above $8$ bounds it gives  $\dsum(\al) \leq  \frac{  \sum_{k=1}^{10} \alpha_k  +  \alpha_{10} -\alpha_{9} }{2}$.

\subsection{Proof for the case with $K=13$} 
When $K=13$, the following bounds are directly derived from  Lemma~\ref{lm:dLsum}
\begin{align}
8 d_1 + 4 d_5 +  2 d_7 + d_{12} +  d_{13} & \leq  4\alpha_1 + 2\alpha_5 + \alpha_7+ \alpha_{13}  \label{eq:uboundgK1301}    \\
8 d_2 + 4 d_5 +  2 d_7+ d_{12} +  d_{13} & \leq  4\alpha_2 + 2\alpha_5 + \alpha_7+ \alpha_{13} \label{eq:uboundgK1302}    \\
8 d_3 +  4 d_6 +  2 d_7+ d_{12} +  d_{13} & \leq 4\alpha_3 + 2\alpha_6 + \alpha_7+ \alpha_{13}  \label{eq:uboundgK1303}    \\
8 d_4 + 4 d_6 +  2 d_7+ d_{12} +  d_{13} & \leq  4\alpha_4 +  2\alpha_6 + \alpha_7+ \alpha_{13}  \label{eq:uboundgK1304}     \\
   4 d_9 +2 d_{11} +  d_{12} +  d_{13}  & \leq    \quad \quad \  \  2\alpha_9  +   \alpha_{11}+ \alpha_{13}  \label{eq:uboundgK1305}    \\
    4 d_9 +2 d_{11}+ d_{12} + d_{13}  & \leq      \quad  \quad   \ \  2\alpha_9   + \alpha_{11}+ \alpha_{13}  \label{eq:uboundgK1306}    \\
    4 d_{10} +  2 d_{11} +    d_{12} +  d_{13} & \leq   \quad \quad \   2\alpha_{10} +   \alpha_{11}+ \alpha_{13}  \label{eq:uboundgK1307}    \\
  8 d_8 +   4 d_{10} +   2d_{11} +    d_{12} +  d_{13} & \leq    4\alpha_8 +   2\alpha_{10} +   \alpha_{11}+ \alpha_{13}.  \label{eq:uboundgK1308}    
\end{align}
The above $8$ bounds reveal that  $\dsum(\al) \leq  \frac{  \sum_{k=1}^{13} \alpha_k  +  \alpha_{13} -\alpha_{12} }{2}$.

\subsection{Proof for the case with $K=16$} 
When $K=16$, the following bounds are directly derived from  Lemma~\ref{lm:dLsum} 
\begin{align}
8 d_1 + 4 d_5 +  2 d_7 + d_{15} +  d_{16} & \leq  4\alpha_1 + 2\alpha_5 + \alpha_7+ \alpha_{16}  \non  \\
8 d_2 + 4 d_5 +  2 d_7+ d_{15} +  d_{16} & \leq  4\alpha_2 + 2\alpha_5 + \alpha_7+ \alpha_{16} \non  \\
8 d_3 +  4 d_6 +  2 d_7+ d_{15} +  d_{16} & \leq 4\alpha_3 + 2\alpha_6 + \alpha_7+ \alpha_{16}  \non  \\
8 d_4 + 4 d_6 +  2 d_7+ d_{15} +  d_{16} & \leq  4\alpha_4 +  2\alpha_6 + \alpha_7+ \alpha_{16}  \non   \\
 8 d_8 +  4 d_{12} +2 d_{14} +  d_{15} +  d_{16}  & \leq     4\alpha_8 +  2\alpha_{12}  +   \alpha_{14}+ \alpha_{16}  \non  \\
  8 d_9 +   4 d_{12} +2 d_{14}+ d_{15} + d_{16}  & \leq      4\alpha_9 + 2\alpha_{12}   + \alpha_{14}+ \alpha_{16}  \non  \\
   8 d_{10} +  4 d_{13} +  2 d_{14} +    d_{15} +  d_{16} & \leq   4\alpha_{10} + 2\alpha_{13} +   \alpha_{14}+ \alpha_{16}  \non  \\
   8 d_{11} +   4 d_{13} +   2d_{14} +    d_{15} +  d_{16} & \leq   4\alpha_{11} +  2\alpha_{13} +   \alpha_{14}+ \alpha_{16} .  \non  
\end{align}
It then implies that  $\dsum(\al) \leq  \frac{  \sum_{k=1}^{16} \alpha_k  +  \alpha_{16} -\alpha_{15} }{2}$.

In the following we will prove Corollary~\ref{cor:ubound} for the general case ($K\geq 3$) by using the result of Lemma~\ref{lm:dLsum}. Note that when  $K=2$, the proof is straightforward. 
\subsection{Proof for the general case }

In our proof, a total of $2^{\Jl}$  bounds are required, which can be seen in the previous examples. 
Among those $2^{\Jl}$  bounds,  the first $2^{\Jl-1}$  bounds have a similar structure. Specifically, when $\ell \in [1 : 2^{\Jl-1}]$,  the $\ell$th bound takes the following form
\begin{align}
  & \sum_{j=0}^{\Jl-1}  2^{\Jl- j}  \cdot d_{ \lceil \ell /2^{j} \rceil  +   \sum^{j}_{l=1} 2^{\Jl-l} }  \  + d_{K-1} +  d_{K}  \non\\ 
  \leq  &  \sum_{j=0}^{\Jl -1}  2^{\Jl- j -1}  \cdot \alpha_{ \lceil \ell /2^{j} \rceil  +   \sum^{j}_{l=1} 2^{\Jl-l} } \  +  \alpha_{K} .    \label{eq:uboundg111}
 \end{align}
 Note that in the above expression, we define that $\sum^{0}_{l=1} 2^{\Jl-l} \defeq 0$.
 When  $\ell\in [2^{\Jl-1}+1 : 2^{\Jl}]$, the $\ell$th bound takes the following form
\begin{align}
 &\sum_{j=0}^{\Jl-1}  2^{\Jl- j}  \cdot d_{\Theta( K- 1 - 2^{\Jl}   +   \lceil (\ell -2^{\Jl-1}) /2^{j} \rceil  +   \sum^{j}_{l=1} 2^{\Jl-l} ) }  \  + d_{K-1} +  d_{K}  \non\\
  \leq &   \sum_{j=0}^{\Jl -1}  2^{\Jl- j -1}  \cdot \alpha_{\Theta( K- 1 - 2^{\Jl}   +   \lceil (\ell -2^{\Jl-1}) /2^{j} \rceil  +   \sum^{j}_{l=1} 2^{\Jl-l} ) } \  +  \alpha_{K}   \label{eq:uboundg222}  
 \end{align}
where $\Theta(\bullet)$, $d_{0}$ and $ \alpha_{0}$ are defined in \eqref{eq:thetadef111}, \eqref{eq:thetadef222} and \eqref{eq:thetadef333}. 
The last $2^{\Jl-1}$ bounds have a similar structure as the first $2^{\Jl-1}$ bounds. However, with our design in \eqref{eq:uboundg222}, we enforce some  $d_{\Theta(\bullet) }$ and $\alpha_{\Theta(\bullet) }$ to $0$ when the corresponding  indices are less than $ 2^{\Jl}$. For example,  when $K=13$ and $\Jl=\lceil \log \frac{K}{2} \rceil =3$, the first $2^{\Jl-1} = 4$ bounds are exactly the same as in \eqref{eq:uboundgK1301}-\eqref{eq:uboundgK1304}, while the last $4$ bounds are expressed as 
\begin{align}
 8 d_{\Theta(5)} +   4 d_9 +2 d_{11} +  d_{12} +  d_{13}  & \leq   4\alpha_{\Theta(5)} +  2\alpha_9  +   \alpha_{11}+ \alpha_{13}  \label{eq:uboundgK1305a}    \\
  8 d_{\Theta(6)} +   4 d_9 +2 d_{11}+ d_{12} + d_{13}  & \leq     4\alpha_{\Theta(6)}+ 2\alpha_9   + \alpha_{11}+ \alpha_{13}  \label{eq:uboundgK1306a}    \\
 8 d_{\Theta(7)} +    4 d_{10} +  2 d_{11} +    d_{12} +  d_{13} & \leq   4\alpha_{\Theta(7)} +   2\alpha_{10} +   \alpha_{11}+ \alpha_{13}  \label{eq:uboundgK1307a}    \\
  8 d_8 +   4 d_{10} +   2d_{11} +    d_{12} +  d_{13} & \leq    4\alpha_8 +   2\alpha_{10} +   \alpha_{11}+ \alpha_{13}  \label{eq:uboundgK1308a}    
\end{align}
where   $d_{\Theta(5)} = d_{\Theta(6)}=d_{\Theta(7)} = \alpha_{\Theta(5)} = \alpha_{\Theta(6)} = \alpha_{\Theta(7)} =0$. The bounds in \eqref{eq:uboundgK1305a}-\eqref{eq:uboundgK1308a} can be rewritten as in 
 \eqref{eq:uboundgK1305}-\eqref{eq:uboundgK1308}.

Note that, for the left-hand side of the above $2^{\Jl}$  bounds, the total weight of $d_k$ is $2^{\Jl}$, $\forall k \in [1: K]$.  For the right-hand side of the above  $2^{\Jl}$  bounds, the total weight of $\alpha_k$ is $2^{\Jl -1}$, $\forall k \in [1: K-2]$;  the  total weight of $\alpha_K$ is $2^{\Jl}$; and the total weight of $\alpha_{K-1}$ is $0$.
Therefore, by summing up the above $2^{\Jl}$ bounds and dividing each side with $2^{\Jl}$, the following bound holds true  \[\dsum(\al) \leq  \frac{  \sum_{k=1}^K \alpha_k  +  \alpha_K -\alpha_{K-1} }{2}   \]
which completes the proof of Corollary~\ref{cor:ubound}.

\section{Proofs of Lemmas~\ref{lm:powerans}, \ref{lm:mindiskl}, \ref{lm:boundTkl}} \label{app:achiev}

Recall that, when  $\ell  \in [1:K-2]$, we have $|\Ic_{k,\ell}| = \me^{\Ke_{\ell}(\Ke_{\ell}-1) } + (\Ke_{\ell}-1) \me^{\Ke_{\ell}(\Ke_{\ell}-1) -1 } -1$,  $|\mathcal{S}_{k,\ell}|= \me^{\Ke_{\ell}(\Ke_{\ell}-1) }$, $\lambda_{\ell} =     \frac{\alpha_{\ell} - \alpha_{\ell-1}}{\Md_{\ell}} - \epsilon$, $ \Md_{\ell}      \defeq 2 \me^{\Ke_{\ell}(\Ke_{\ell}-1) } + (\Ke_{\ell}-1) \me^{\Ke_{\ell}(\Ke_{\ell}-1) -1 } -1 $,  $\Nd_{\ell} =  \me^{\Ke_{\ell}(\Ke_{\ell}-1) }$, and $\Ke_{\ell} = K-\ell+1$.

\subsection{Proof of Lemma~\ref{lm:powerans}} \label{app:powerans}

Based on the signal design in \eqref{eq:xtran9900}-\eqref{eq:defvelln},  the average power of the transmitted signal at  Transmitter~$k$, $k\in [1:K]$, is bounded by 
  \begin{align}
\E|x_{k}|^2 =& \sum_{\ell=1}^{k}  P^{  - \alpha_{\ell-1}}   \E| x_{k,\ell} |^2      \non\\
 =& \sum_{\ell=1}^{k}  P^{  - \alpha_{\ell-1}}   \E| \vv_{k,\ell}^\T   \bv_{k,\ell} |^2      \non\\
    =& \sum_{\ell=1}^{k}  P^{  - \alpha_{\ell-1}}     \sum_{i=1}^{\Nd_{\ell}} |v_{k,\ell, i} |^2 \cdot \E| b_{k,\ell, i} |^2       \label{eq:avpowe6264}\\
        =& \sum_{\ell=1}^{k}  P^{  - \alpha_{\ell-1}}     \sum_{i=1}^{\Nd_{\ell}} |v_{k,\ell, i} |^2 \cdot  \frac{  \gamma^2  Q_{\ell}(Q_{\ell}+1)}{3Q_{\ell}^2}    \label{eq:avpowe8285}\\
       \leq  & \sum_{\ell=1}^{k}  P^{  - \alpha_{\ell-1}}     \sum_{i=1}^{\Nd_{\ell}} |v_{k,\ell, i} |^2 \cdot \gamma^2    \label{eq:avpowe2735}\\ 
       \leq  &\gamma^2  \sum_{\ell=1}^{k}     \sum_{i=1}^{\Nd_{\ell}} |v_{k,\ell, i} |^2    \label{eq:avpowe2355}\\ 
       \leq  &\gamma^2  \sum_{\ell=1}^{k^{\star}}     \sum_{i=1}^{\Nd_{\ell}} |v_{k^{\star},\ell, i} |^2    \non\\ 
       =  &\gamma^2    \eta   \non
         \end{align}
   where \[ k^{\star}  \defeq  \argmax_{k' \in[1:K]} \sum_{\ell=1}^{k'}     \sum_{i=1}^{\Nd_{\ell}} |v_{k',\ell, i} |^2   \]     
        and \[\eta \defeq \sum_{\ell=1}^{k^{\star}}     \sum_{i=1}^{\Nd_{\ell}} |v_{k^{\star},\ell, i} |^2 .   \] Note that $\eta$ is a positive value  independent of $P$.
    The step in \eqref{eq:avpowe6264} uses the fact that the symbols $\{ b_{k,\ell, i}\}_{k, \ell, i}$ are mutually independent, based on our signal design.
    The step in \eqref{eq:avpowe8285} is from the result of  \eqref{eq:avpowe2233}, given that  $ b_{k,\ell, i}  \in    \Omega ( \xi =  \gamma  \cdot \frac{ 1}{ Q_{\ell} } ,   \   Q =   Q_{\ell}  )$, for $i \in [1:\Nd_{\ell} ]$,    $\ell \in [1 : k]$,  $k \in [1 : K]$ (see \eqref{eq:cons0295}).  
  The step in \eqref{eq:avpowe2735} uses the identity that   $   \frac{Q_{\ell}(Q_{\ell}+1)}{3Q_{\ell}^2}  \leq    \frac{2Q_{\ell}^2}{3Q_{\ell}^2} < 1$.   
  The step in \eqref{eq:avpowe2355} follows from the fact that    $P^{  - \alpha_{\ell-1}}  \leq 1$ for $\ell \in [1 : K]$. 
  At this point, we complete the proof of Lemma~\ref{lm:powerans}.

\subsection{Proof of Lemma~\ref{lm:mindiskl}} \label{app:mindiskl}

Since the  elements of $\mathcal{S}_{k,\ell}$ and $\Ic_{k,\ell}$  are monomials generated from the channel coefficients (see \eqref{eq:defvellnm11} and \eqref{eq:defvelln0099}), the minimum distance $d_{\min}(k, \ell)$ defined in \eqref{eq:mindis7788} can be bounded by using the Khintchine-Groshev Theorem for Monomials (see Theorem~\ref{thm:KG}). Specifically, the Khintchine-Groshev Theorem for Monomials reveals that,  for any small enough $\epsilon' =\epsilon >0$, and for almost all  realizations of   channel coefficients, there exists a positive constant $\kappa$ such that
    \begin{align}
 d_{\min}(k, \ell)  & \geq  \frac{\kappa \gamma \sqrt{P^{ \alpha_{k} - \alpha_{\ell-1} - \lambda_{\ell}}}}{ ( \Ke_{\ell}Q_{\ell})^{| \mathcal{S}_{k,\ell}| + | \Ic_{k,\ell}|   -1 + \epsilon}}   \non\\
 &=   \frac{\kappa   \gamma P^{  ( \alpha_{k} - \alpha_{\ell-1})/2 }}{ P^{  \lambda_{\ell}/2} \cdot ( \Ke_{\ell}  P^{ \lambda_{\ell}/2})^{\Md_{\ell} -1+\epsilon}}  \non \\
 &=  \frac{\kappa   \gamma }{   \Ke_{\ell}^{\Md_{\ell} -1 +\epsilon}}    \cdot    \frac{ P^{  ( \alpha_{k} - \alpha_{\ell-1})/2 }}{ (  P^{ \lambda_{\ell}/2})^{\Md_{\ell} +\epsilon}}   \non\\
  &=  \frac{\kappa   \gamma }{   \Ke_{\ell}^{\Md_{\ell} -1 +\epsilon}}    \cdot    \frac{  P^{  \frac{ \alpha_{k} - \alpha_{\ell-1} -( \alpha_{\ell} - \alpha_{\ell-1} )}{2} }}{ P^{ - \frac{\epsilon}{2} \cdot (\Md_{\ell} +\epsilon - \frac{\alpha_{\ell} - \alpha_{\ell-1}}{\Md_{\ell}})}}    \non \\
    &=    \kappa'  P^{   \frac{ \alpha_{k} -  \alpha_{\ell}  + \epsilon_{\ell} }{2} }  \label{eq:mindis6655}
 \end{align}
for $k\in [\ell : K]$, $\ell  \in [1:K-2]$, where $\epsilon_{\ell}$ and $\kappa'$ are defined as \[\epsilon_{\ell} \defeq   \epsilon(\Md_{\ell} +\epsilon - \frac{\alpha_{\ell} - \alpha_{\ell-1}}{\Md_{\ell}}), \quad \kappa' \defeq  \frac{\kappa   \gamma }{   \Ke_{\ell}^{\Md_{\ell} -1 +\epsilon}}  .\]  Note that the value of  $\kappa' $ is positive and independent of $P$, and  $\epsilon_{\ell}$ is positive, $\forall \ell  \in [1:K-2]$, given that $\epsilon >0$. 
It then completes the proof of Lemma~\ref{lm:mindiskl}. 

\subsection{Proof of Lemma~\ref{lm:boundTkl}} \label{app:boundTkl}

For the term $T_{k,\ell}$ defined in \eqref{eq:SITdef000}, it can be bounded by 
  \begin{align}
  T_{k,\ell} =&     \sum_{l=\ell+1}^{K}  \sum_{j=l}^{K}      \sqrt{P^{  \alpha_{k}- \alpha_{l-1}}}   h_{kj} \vv_{j,l}^\T   \bv_{j,l} \non\\
   = &   \sum_{l=\ell+1}^{K}  \sum_{j=l}^{K}   \sqrt{P^{  \alpha_{k}- \alpha_{l-1}}}   h_{kj}    \sum_{i=1}^{\Nd_{l}}  v_{j,l, i}    b_{j,l, i}  \non\\
    \leq  &  \sum_{l=\ell+1}^{K}  \sum_{j=l}^{K}  \sqrt{P^{  \alpha_{k}- \alpha_{l-1}}}   |h_{kj}|   \sum_{i=1}^{\Nd_{l}}  |v_{j,l, i}|  \gamma   \label{eq:boundTkl2355}\\
    \leq  &  \sum_{l=\ell+1}^{K}  \sum_{j=l}^{K}   \sqrt{P^{  \alpha_{k}- \alpha_{\ell}}}   |h_{kj}|   \sum_{i=1}^{\Nd_{l}}  |v_{j,l, i}|  \gamma  \non\\
     =  &  \sqrt{P^{  \alpha_{k}- \alpha_{\ell}}}\cdot \gamma   \sum_{l=\ell+1}^{K}  \sum_{j=l}^{K}    \sum_{i=1}^{\Nd_{l}}  |h_{kj}|    |v_{j,l, i}|    \non\\
        =  &  \sqrt{P^{  \alpha_{k}- \alpha_{\ell}}}  \cdot \delta_{k,\ell}    \non
        \end{align}
for $k \in [\ell : K]$, $\ell \in [1 : K-2]$,  where \[\delta_{k,\ell} \defeq \gamma  \sum_{l=\ell+1}^{K}  \sum_{j=l}^{K}      \sum_{i=1}^{\Nd_{l}}  |h_{kj}|    |v_{j,l, i}|     \] and the value of $\delta_{k,\ell}$  is  independent of $P$.
The step in \eqref{eq:boundTkl2355} uses the fact that $b_{j,\ell, i} \leq \gamma $, given that  $ b_{k,\ell, i}  \in    \Omega ( \xi =  \gamma  \cdot \frac{ 1}{P^{ \frac{ \lambda_{\ell} }{2}}} ,   \   Q =  P^{ \frac{ \lambda_{\ell} }{2}} )$, for $i \in [1:\Nd_{\ell} ]$,    $k \in [\ell : K]$,  $\ell \in [1 : K]$ (see \eqref{eq:cons0295}).
At this point, we complete the proof of Lemma~\ref{lm:boundTkl}.


\end{document}